\documentclass[acmsmall,10pt]{acmart}

\renewcommand\footnotetextcopyrightpermission[1]{} 
\acmConference{}{}{}
\acmYear{}
\acmISBN{} 
\acmDOI{} 
\startPage{1}

\setcopyright{none}

\usepackage{booktabs}   
\usepackage{subcaption} 
\usepackage{paralist}
\usepackage{cleveref}
\usepackage{graphicx,color}
\usepackage{listing}
\usepackage{listings}
\usepackage{mathpartir}
\usepackage{mdframed}
\usepackage{tikz}
\usetikzlibrary{arrows, fit, positioning, shapes}
\usepackage{xcolor}
\usepackage{xspace}

\NewEnviron{centerframe}[1][]{%
    \begin{mdframed}
		\begin{center}
		\BODY
		\end{center}
		\end{mdframed}
}

\newcommand{\commentout}[1]{}

\newcommand{\tup}[1]{\left<#1\right>\xspace}
\newcommand{\set}[1]{\left\{#1\right\}\xspace}
\newcommand{\brackets}[1]{\left[#1\right]\xspace}
\newcommand{\dom}[1]{dom({#1})\xspace}
\newcommand{\img}[1]{img({#1})\xspace}

\newcommand{\pfun}{\rightharpoonup}

\newcommand{\true}{true}
\newcommand{\false}{false}

\newcommand{\move}{Move\xspace}
\newcommand{\libra}{Libra\xspace}
\newcommand{\redacted}[1]{#1}

\newcommand{\mem}{\ensuremath{{\mathcal{M}}}\xspace}
\newcommand{\cl}{\ensuremath{\mathcal{L}}}
\newcommand{\rn}{\ensuremath{R}}
\newcommand{\globalstates}{\ensuremath{{GS}}\xspace}
\newcommand{\globalresid}{\ensuremath{\mathcal{G}}}

\newcommand{\varid}{\ensuremath{\mathcal{V}}}
\newcommand{\localvalues}{\ensuremath{{ LV}\xspace}}
\newcommand{\tagval}{\ensuremath{TV}}
\newcommand{\val}{\ensuremath{V}}
\newcommand{\references}{\ensuremath{Ref}}
\newcommand{\tg}{\ensuremath{T}}

\newcommand{\prival}{\ensuremath{P}}

\newcommand{\fn}{\ensuremath{F}}

\newcommand{\stackvalues}{\ensuremath{{{ SV}}\xspace}}

\newcommand{\structname}{\ensuremath{\mathcal{T}}}

\newcommand{\accadd}{\ensuremath{\mathcal{A}}}
\newcommand{\unres}{\ensuremath{{\bf U}\xspace}}

\newcommand{\mut}{\ensuremath{{\bf mut}\xspace}}
\newcommand{\immut}{\ensuremath{{\bf immut}\xspace}}

\newcommand{\legal}{\textit{legal}\,}

\newcommand{\resourcesof}[1]{\mathcal{R}({#1})\xspace}
\newcommand{\resourcesi}[1]{\mathcal{R}_{I}({#1})\xspace}
\newcommand{\resourcese}[1]{\mathcal{R}_{E}({#1})\xspace}

\newcommand{\tv}[2]{\tup{{#1},{#2}}\xspace}
\newcommand{\record}[1]{\set{#1}\xspace}
\newcommand{\typeof}[2]{{#1}:{#2}}
\newcommand{\cons}{{::}}
\newcommand{\stackht}[2]{{#1}\cons{#2}\xspace}

\newcommand{\rff}{\ensuremath{{\bf ref}}\xspace}

\newcommand{\rft}[3]{\rff~\tup{{#1}, {#2}, {#3}}\xspace}

\newcommand{\st}[1]{\tup{#1}\xspace}

\newcommand{\ao}[2]{\ensuremath{{#1}\tup{#2}}\xspace}

\newcommand{\mdel}[2]{{#1}\setminus{#2}}
\newcommand{\mset}[3]{{#1}\brackets{{#2} \mapsto {#3}}}
\newcommand{\mread}[2]{#1(#2)\xspace}

\newcommand{\lset}[3]{{#1}\brackets{{#2} \mapsto {#3}}}
\newcommand{\lread}[2]{#1(#2)\xspace}

\newcommand{\gdel}[2]{{#1}\setminus{#2}}
\newcommand{\gset}[3]{{#1}\brackets{{#2} \mapsto {#3}}}
\newcommand{\gread}[2]{#1(#2)\xspace}

\newcommand{\cmdname}[1]{{\ensuremath{\bf {#1}}}\xspace}
\newcommand{\rulename}[1]{{\ensuremath{\text{\bf {#1}}}}\xspace}

\newcommand{\movelocCmd}{\cmdname{MvLoc}}
\newcommand{\movelocRule}{\rulename{MvLoc}}
\newcommand{\movelocRefRule}{\rulename{MvLoc-Ref}}
\newcommand{\copylocCmd}{\cmdname{CpLoc}}
\newcommand{\copylocvalRule}{\rulename{CpLoc}}
\newcommand{\copylocrefRule}{\rulename{CpLoc-Ref}}
\newcommand{\storelocCmd}{\cmdname{StLoc}}

\newcommand{\storelocbotrefRule}{\rulename{StLoc}}
\newcommand{\borrowlocCmd}{\cmdname{BorrowLoc}}
\newcommand{\borrowlocRule}{\rulename{BorrowLoc}}
\newcommand{\readrefCmd}{\cmdname{ReadRef}}
\newcommand{\readrefRule}{\rulename{ReadRef}}
\newcommand{\writerefCmd}{\cmdname{WriteRef}}
\newcommand{\writerefRule}{\rulename{WriteRef}}
\newcommand{\freezerefCmd}{\cmdname{FreezeRef}}
\newcommand{\freezerefRule}{\rulename{FreezeRef}}
\newcommand{\callCmd}{\cmdname{Call}}

\newcommand{\returnCmd}{\cmdname{Ret}}

\newcommand{\packCmd}{\cmdname{Pack}}
\newcommand{\packResourceRule}{\rulename{Pack-R}}
\newcommand{\packUnrestrictedRule}{\rulename{Pack-U}}

\newcommand{\unpackCmd}{\cmdname{Unpack}}
\newcommand{\unpackRule}{\rulename{Unpack}}

\newcommand{\borrowfieldCmd}{\cmdname{BorrowField}}
\newcommand{\borrowfieldRule}{\rulename{BorrowField}}

\newcommand{\movetoCmd}{\cmdname{MoveTo}}
\newcommand{\movetoRule}{\rulename{MoveTo}}

\newcommand{\movefromCmd}{\cmdname{MoveFrom}}
\newcommand{\movefromRule}{\rulename{MoveFrom}}

\newcommand{\borrowglobalCmd}{\cmdname{BorrowGlobal}}
\newcommand{\borrowglobalRule}{\rulename{BorrowGlobal}}

\newcommand{\existsCmd}{\cmdname{Exists}}
\newcommand{\existsRule}{\rulename{Exists}}
\newcommand{\popCmd}{\cmdname{Pop}}
\newcommand{\popValRule}{\rulename{Pop}}
\newcommand{\popRefRule}{\rulename{Pop-Ref}}

\newcommand{\loadconstCmd}{\cmdname{LoadConst}}
\newcommand{\loadconstRule}{\rulename{LoadConst}}
\newcommand{\stackopCmd}{\cmdname{Op}}
\newcommand{\stackopRule}{\rulename{StackOp}}
\newcommand{\stackop}{\cmdname{Op}}

\newcommand{\branchCmd}{   \cmdname{Branch}}

\newcommand{\op}{\ensuremath{{\mathrm op}}}
\newcommand{\pc}{\ensuremath{{\mathrm pc}}}
\newcommand{\pcstate}[2]{\ensuremath{\tup{{#1},{#2}}\xspace}}
\newcommand{\pcstep}[3]{\ensuremath{{#1}\vdash{#2}\rightarrow{#3}\xspace}}
\newcommand{\pcrule}[5][]{\ensuremath{\inferrule*[Right=\rulename{#1}]{#2}{\pcstep{#3}{#4}{#5}}}\xspace}

\newcommand{\steppcRule}{\cmdname{Step}}
\newcommand{\branchtruepcRule}{\rulename{Branch-T}}
\newcommand{\branchfalsepcRule}{\rulename{Branch-F}}

\newcommand{\progloc}[2]{{#1}\brackets{#2}}

\newcommand{\sep}{\quad}
\newcommand{\step}[3]{\tup{{#1},{#2}}\!\rightarrow_{0}\!{#3}\xspace}
\newcommand{\semrule}[5][]{\ensuremath{\inferrule*[Right=\rulename{#1}]{#2}{\step{#3}{#4}{#5}}}\xspace}

\definecolor{pblue}{rgb}{0.13,0.13,1}
\definecolor{pgreen}{rgb}{0,0.5,0}
\definecolor{pred}{rgb}{0.9,0,0}
\definecolor{pgrey}{rgb}{0.46,0.45,0.48}

\definecolor{ckeyword}{HTML}{7F0055}
\definecolor{ccomment}{HTML}{3F7F5F}
\definecolor{cnumber}{HTML}{2A0099}

\lstdefinelanguage{Solidity}{
  keywords={contract, function, payable, event, msg},
  ndkeywords={address, uint},
  showspaces=false,
  showtabs=false,
  breaklines=true,
  showstringspaces=false,
  breakatwhitespace=true,
  lineskip=-0.6pt,
  morecomment=[l]{//}, 
  morecomment=[s]{/*}{*/}, 
  basewidth={0.48em, 0.4em},%
  basicstyle=\scriptsize\ttfamily,
  keywordstyle={\color{ckeyword}\ttfamily\bfseries},
  ndkeywordstyle={\color{pblue}\ttfamily\bfseries},
  commentstyle={\color{ccomment}\itshape},
  stringstyle=\color{green},
  moredelim=[il][\textcolor{pgrey}]{$$},
  moredelim=[is][\textcolor{pgrey}]{\%\%}{\%\%}
}

\lstdefinelanguage{Scilla}{
  keywords={accept, contract, emit, end, field, transition, event, send, _sender, _amount, _recipient, match, with},
  ndkeywords={Address, Map, ByStr20, Uint},
  showspaces=false,
  showtabs=false,
  breaklines=true,
  showstringspaces=false,
  breakatwhitespace=true,
  lineskip=-0.6pt,
  morecomment=[l]{//}, 
  morecomment=[s]{/*}{*/}, 
  basewidth={0.48em, 0.4em},%
  basicstyle=\scriptsize\ttfamily,
  keywordstyle={\color{ckeyword}\ttfamily\bfseries},
  ndkeywordstyle={\color{pblue}\ttfamily\bfseries},
  commentstyle={\color{ccomment}\itshape},
  stringstyle=\color{green},
  moredelim=[il][\textcolor{pgrey}]{$$},
  moredelim=[is][\textcolor{pgrey}]{\%\%}{\%\%}
}

\lstdefinelanguage{Move}{
 keywords={
abort, acquires, assert, copy, borrow_global, borrow_global_mut, create_account, fun, freeze, module, move_to, move_from, public, resource, break, use, vector
    continue,
    else,
    false,
    if,
    let, loop,
    move, mut,
    return,
    struct,
    unsafe,
    true,
    while,
  },
  ndkeywords={address, bool, u64, bytearray, Self},
  showspaces=false,
  showtabs=false,
  breaklines=true,
  showstringspaces=false,
  breakatwhitespace=true,
  lineskip=-0.6pt,
  morecomment=[l]{//}, 
  morecomment=[s]{/*}{*/}, 
  basewidth={0.48em, 0.4em},%
  basicstyle=\scriptsize\ttfamily,
  keywordstyle={\color{ckeyword}\ttfamily\bfseries},
  ndkeywordstyle={\color{pblue}\ttfamily\bfseries},
  commentstyle={\color{ccomment}\itshape},
  stringstyle=\color{green},
  moredelim=[il][\textcolor{pgrey}]{$$},
  moredelim=[is][\textcolor{pgrey}]{\%\%}{\%\%}
}

\lstdefinestyle{number}{%
  numbers=left,%
  numberstyle=\scriptsize\em,%
  xleftmargin=1em%
}

\newcommand{\solcode}[1]{\lstinline[language=Solidity,basicstyle=\small\ttfamily]{#1}}
\newcommand{\scilcode}[1]{\lstinline[language=Scilla,basicstyle=\small\ttfamily]{#1}}
\newcommand{\mcode}[1]{\lstinline[language=Move,basicstyle=\small\ttfamily]{#1}}
\newcommand{\code}[1]{\mcode{#1}}
\newcommand{\smallcode}[1]{\lstinline[language=Move,basicstyle=\scriptsize\ttfamily]{#1}}

\begin{document}

\title{Resources: A Safe Language Abstraction for Money}


\author{Sam Blackshear}
\affiliation{
  \institution{Novi}
}
\email{shb@fb.com}

\author{David L. Dill}
\affiliation{
  \institution{Novi}
}
\email{dill@fb.com}

\author{Shaz Qadeer}
\affiliation{
  \institution{Novi}
}
\email{shaz@fb.com}

\author{Clark W. Barrett}
\affiliation{
  \institution{Stanford University}
}
\email{barrett@cs.stanford.edu}

\author{John C. Mitchell}
\affiliation{
  \institution{Stanford University}
}
\email{john.mitchell@stanford.edu}

\author{Oded Padon}
\affiliation{
  \institution{Stanford University}
}
\email{padon@cs.stanford.edu}

\author{Yoni Zohar}
\affiliation{
  \institution{Stanford University}
}
\email{yoniz@cs.stanford.edu}

\begin{abstract}
\emph{Smart contracts} are programs that implement potentially sophisticated
transactions on modern blockchain platforms. In the rapidly evolving blockchain environment,
smart contract programming languages must allow users to write expressive programs that manage and transfer assets,
yet provide strong protection against sophisticated attacks.
Addressing this need, we present flexible and reliable abstractions for programming with digital currency in the \move language \cite{move_white}.
\move uses novel linear~\cite{linear_logic} \emph{resource} types with semantics
drawing on C++11~\cite{c++} and Rust~\cite{rust}: 
when a resource value is assigned to a new memory location, the location previously holding it must be invalidated.
In addition, a resource type can only be created or destroyed by procedures inside its declaring module.
We present an executable bytecode language with resources and prove that it enjoys \emph{resource safety}, a conservation property for program values that is analogous to conservation of mass in the physical world.

\end{abstract}



\maketitle

\section{Introduction}
\label{sec:introduction}

The emergence of Bitcoin \cite{nakamoto} and Ethereum \cite{ethereum} has created significant interest in
the computational model of a replicated state machine synchronized by a distributed consensus protocol.
In this programming model, a command is executed as an atomic
and deterministic transaction that is replicated consistently across all nodes participating in consensus.
While cryptocurrency and decentralized finance are the most prominent applications of programmable blockchains,
there are other important use-cases such as tracking supply chains \cite{HBR-global-supply-chain}
and clearing global markets \cite{JPMcoin}.


Transactions are programmed as \emph{smart contracts,}
a catchy name \cite{szabo_smart_contracts} for program units installed for atomic execution on the blockchain.
If the contract language is sufficently expressive, then
smart contracts are attractive implementaions for a wide variety of conventional
functions such as bank deposit and withdrawal, cross-border funds transfer, point-of-sale online payment,
escrow agreements, futures contracts, and derivatives.
To meet these goals, a smart-contract programming language must allow users to write programs that
manage and transfer assets while providing extremely trustworthy protection against sophisticated attacks.

In this paper, we describe and analyze flexible and reliable abstractions for programming with
digital currency and other assets in the \move language \cite{move_white}.
\move uses novel linear~\cite{linear_logic} \emph{resource} types that draw on
experience with C++11~\cite{c++} and Rust~\cite{rust}
to preserve integrity and prevent copying of assets.
When combined with other abstraction features of \move, linearity ensures resource conservation.
Whereas data abstraction ensures that
a resource may only be created and destroyed by the defining module,
linearity further prevents duplication and unintended loss.
We present an executable \move{} bytecode language with move semantics and
show that it satisfies a set of \emph{resource safety} guarantees.

\paragraph{Contributions}
This paper adds rigor to the informal description of \move \cite{move_white}.
Its key contributions are:
\begin{itemize}
\item We introduce \emph{resources}, an intuitive abstraction for currency-like values, and demonstrate their utility compared to existing language constructs (\Cref{sec:overview}).
\item We explain the key features of the \move bytecode language and explain how their design supports support resource-oriented programming (\Cref{sec:move_overview})
\item We formalize the semantics of the \move bytecode interpreter for the subset of \move{} analyzed in this paper (\Cref{sec:bytecode_formal}).
\item We formally define \emph{resource safety} properties and prove that execution of \move bytecode programs is resource-safe (\Cref{sec:resource_safety}).
\item We describe our implementation of the \move virtual machine, its integration in the \libra blockchain~\cite{libra_blockchain_white}, and the adoption of \move in other contexts (\Cref{sec:experience}).
\end{itemize}

\section{Programming With Money}
\label{sec:overview}

\move{} is designed to support a rich variety of economic and financial activities by supporting fundamental conservation properties,
not only for built-in currencies, but also for programmer-defined assets.
We believe this is essential.
To begin with, smart contracts provide customizable logic for sending, receiving, storing, and apportioning digital funds
that cannot be arbitrarily created, lost, or destroyed.
Further, the internal balance in a bank account, the monetary value inherent in a contract for future payment, or an escrow contract
all represent assets that must be conserved in the same ways as conventional currency.
Thus, smart contracts must be able to implement new assets with expected conservation properties and appropriately control the exchange of one asset for another.

\subsection{Savings Bank Example}
With this goal in mind, we use a simple bank account contract to illustrate the key features of \move for programming with assets and
demonstrate by example the advantages of \move over two alternative contract programming languages where notable problems have occurred in practice.
\Cref{fig:bank} implements a savings bank with the following requirements:
\begin{itemize}
\item A customer should be able to deposit money worth $N$ via the \code{deposit} procedure and subsequently extract money worth $N$ via the \code{withdraw} procedure.
\item No customer should be able to withdraw money deposited by another customer.
\end{itemize}

Even in this simplest of examples, there are already \emph{two} assets: the funds deposited into the bank contract, and the bank credit that the customer can use to withdraw the funds in the future.
Most smart contract platforms have a \emph{native asset} such as \emph{Ether} in Ethereum \cite{ethereum} that is implemented as part of the core platform and guarantees conservation.
But even if the deposited funds are represented using the native asset, the bank contract must correctly implement \code{deposit} and \code{withdraw} to ensure conservation for the bank credit asset.
Programming mistakes in this setting can be extremely costly;high-profile bugs in Ethereum, \emph{e.g.,} \cite{re_dao, parity_hack,eth_vulns}, have resulted in the theft of digital assets worth tens of millions of dollars.
To summarize, programming challenges in this environment include:
\begin{enumerate}
\item \textbf{Conservation.} Transfers must preserve the total supply of money in the system, including custom assets defined by contracts.

\item \textbf{Unique atomic transfer.} The sender of an asset must relinquish all control of the asset.
This ownership transfer should be \emph{atomic} because any non-atomic exchange risks leaving one or both parties empty-handed.

\item \textbf{Authority.} Smart contract programmers must represent authority carefully and restrict access to privileged operations. Contracts are deployed on a public platform open to both benign customers and bad actors.
\end{enumerate}

\move represents money using user-defined linear \emph{resource} types.
\move has ordinary types like integers and addresses that can be copied, but resources can only be moved.
Linearity prevents ``double spending'' by moving a resource twice (e.g., into two different callees)
and forces a well-typed procedure to move all of its resources, avoiding accidental loss.

\Cref{fig:bank} provides a \move representation of the simple bank along with an implementation in Solidity~\cite{solidity} and Scilla~\cite{scilla}.
Solidity is a source language for Ethereum~\cite{ethereum} and the first to provide an expressive smart contract programming model.
Scilla is a newer language designed by programming language researchers to simplify formal verification of contracts
and incorporate lessons learned from Solidity design flaws.
Although many other contract languages have been proposed (see \Cref{sec:related_work}), these two
represent the state of practice (Solidity) and the state of the art (Scilla).

Solidity, Scilla and other \emph{account-based} languages often use a model in which each contract has an implicit balance in the platform's native currency.
This balance can only be modified by special instructions. However, the properties ensured by these special instructions
are not available to programmers that wish to implement custom currencies such as bank credits.
A common strategy used instead is illustrated in \Cref{fig:bank}:
a map, \code{credit}, is employed to map creditor identities to integers.
The integers in the range of the map represent money and must be manipulated carefully to provide the global conservation invariants
associated with monetary assets. However, as we will see by examining the code samples, properties guaranteed by construction in \move are more difficult to ensure via ad hoc programming in other languages.
Although the bank is a somewhat artificial example, it is adapted from similar examples in the Solidity/Scilla documentation and concisely captures the key idioms of typical contracts: sending/receiving/atomically exchanging money and implementing a new money-like construct.

\begin{figure*}[t]
\centering
\begin{tabular}{ccc}
{
\begin{lstlisting}[basicstyle=\scriptsize\ttfamily, language=Solidity]
contract Bank
mapping (address => uint) credit;

function deposit() payable {
  amt =
    credit[msg.sender] + msg.value
  credit[msg.sender] = amt
}

function withdraw() {
  uint amt = credit[msg.sender];
  msg.sender.transfer(amt);
  credit[msg.sender] = 0;
}
\end{lstlisting}
}&{
\begin{lstlisting}[basicstyle=\scriptsize\ttfamily,language=Scilla]
contract Bank
field credit: Map Address Uint;

transition deposit()
  accept;
  match credit[_sender] with
  Some(amt) =>
    credit[_sender] :=
      amt + _amount
  None =>
    credit[_sender] := _amount
  end
end

transition withdraw()
  match credit[_sender] with
  Some(amt) =>
    msg = {
      _recipient: _sender;
      _amount: amt
    };
    credit[_sender] := 0;
    send msg
  None => ()
  end
end
\end{lstlisting}
}&{
\begin{lstlisting}[basicstyle=\scriptsize\ttfamily,language=Move]
module Bank
use 0x0::Coin;
resource T { balance: Coin::T }
resource Credit { amt: u64, bank: address }

fun deposit(
  coin: Coin::T,
  bank: address
): Credit {
  let amt = Coin::value(&coin);
  let t = borrow_global<T>(copy bank);
  Coin::deposit(&mut t.balance, move coin);
  return Credit {
    amt: move amt, bank: move bank
  };
}

fun withdraw(credit: Credit): Coin::T {
  Credit { amt, bank } = move credit;
  let t = borrow_global<T>(move bank);
  return Coin::withdraw(
    &mut t.balance, move amt
  );
}
\end{lstlisting}
}
\end{tabular}

\caption{A simple bank contract in Solidity (left), Scilla (middle), and \move (right). Each code snippet must implement bidirectional exchanges of the language's native currency for a bank credit currency defined by the contract. In Solidity and Scilla, both native and custom currencies are represented indirectly via maps of identities to integers, whereas in \move, currency is represented directly with resources.}
\label{fig:bank}
\end{figure*}

\paragraph{Solidity and Scilla \code{deposit}}
The first task of the \code{deposit} procedure in \Cref{fig:bank} is to accept the language's native currency.
In Solidity, native currency sent by a caller is implicitly deposited into the contract's balance before the callee code is executed,
provided the receiving function is marked as \solcode{payable}.
If not, an attempted deposit causes a runtime failure that reverts all changes performed by the current transaction.

In Scilla, money is transferred from caller to callee via an explicit \scilcode{accept} construct which
avoids runtime failures but introduces other problems.
Although money not accepted by the callee will be silently returned,
bugs may occur if the programmer forgets to \scilcode{accept} funds.
For example, accepting on one control-flow path but not another (e.g., only in the \scilcode{None} branch)
would allow the caller to steal funds deposited by another user by subsequently invoking \scilcode{Withdraw}.

The second task of \code{deposit} is to update the caller's bank credits by the transferred amount.
In both languages, the amount sent by the caller is available through special integer-typed expressions:
\solcode{msg.value} in Solidity and \scilcode{_amount} in Scilla.
The identity of the caller is represented by \solcode{msg.sender} or \scilcode{_sender}, respectively.
\emph{The programmer must be careful to increment the caller's credit balance by the transferred quantity exactly once.
Forgetting to update the balance is stealing funds from the caller,
whereas updating more than once allows the caller to steal funds from other customers.}
There are no special checks on integer expressions to prevent either programmer
error from violating conservation of funds.

\paragraph{Solidity and Scilla \code{withdrawal}}
The \code{withdraw} procedure exchanges bank credits for native currency.
Although this is logically the inverse of \code{deposit}, the implementation looks quite different.
This is because Solidity and Scilla do not have language support for returning native or custom currency to the calling procedure.
Instead, the code uses language primitives for sending currency to the \emph{address} that stores the contract whose procedure invoked \code{withdraw}.

In Solidity, the relevant primitive is \solcode{msg.sender.transfer}.
Subtly, this is a virtual call that invokes a user-defined procedure known as a \emph{fallback function} in the callee.
The decision to make every payment of native currency a virtual call has led to infamous \emph{re-entrancy} vulnerabilities such as the DAO \cite{re_dao} attack that led to theft of digital assets worth over \$60 million.
\emph{The key issues are that (a) the update to the credit map via \solcode{credit[msg.sender] = 0} and the sending of funds via \solcode{transfer} are not atomic, and (b) the map update occurs after the virtual \solcode{transfer} call.}
If the virtual call invokes a user-defined function that calls back into \code{withdraw}, the caller can steal funds deposited by a different customer.

Scilla improves on Solidity by defining a more restricted message-passing primitive for sending money to addresses.
The \scilcode{_amount: amt} code snippet implicitly withdraws \code{amt} units
of money from the contract's available balance.
Then, the \scilcode{_sender}'s balance in the credit map is zeroed out before using the \scilcode{send} primitive to transfer the money to its recipient.
Scilla's type system forces any global side effect like a message \scilcode{send} to occur at the end of the procedure; for our example, it would not allow the update to \scilcode{credit} to occur after the \scilcode{send}.
In addition, Scilla does not have virtual calls. These restrictions prevent re-entrancy issues.

However, the Scilla design introduces a new kind of issue:
\emph{using \scilcode{emit msg} instead of \scilcode{send msg} in the example would cause the money in the message to be destroyed.}
The \scilcode{emit} construct emits the message as a client-facing event rather than sending it to an address.
This mistake permanently reduces the supply of money in the system.
Scilla programmers have encountered this problem in practice (\cite{scilla}, Section 5.2),
though Scilla has an auxiliary ``cashflow'' static analyzer for detecting problems like this.

\paragraph{\move Bank}
The \move implementations of the \code{deposit} and \code{withdraw} procedures are symmetric.
The \code{deposit} procedure says that it requires payment by declaring a parameter of type \code{Coin::T} and that it intends to credit the caller by declaring a return value of type \code{Bank::Credit}.
The \code{withdraw} procedure does the inverse.
\code{Coin::T} represents native currency; it is a resource type defined in a separate \code{Coin} module that we describe in \Cref{sec:core_modules}.
Since both \code{Coin::T} and \code{Bank::Credit} are resources, the type system will reject any implementation that fails to consume the input resource or return ownership of the output resource.
Both resources can leverage the same language feature (move semantics) for atomic ownership transfer into and out of the procedure.

The \code{deposit} code consumes its input resource by acquiring a reference to a \code{Bank::T} value published in global storage and moving the \code{coin} resource into the bank's balance via the call to \code{Coin::deposit}.
It then \emph{packs} (constructs) a \code{Credit} resource and returns it to the caller.

The \code{withdraw} code consumes its input \code{Credit} resource by \emph{unpacking} it.
Unpacking destroys a resource and returns its contents.
Only the \code{Bank} module can pack, unpack, and acquire references to the fields of the \code{Credit} resource; code outside the module can only access \code{Credit} through the procedures exposed by \code{Bank}.
Finally, \code{withdraw} extracts native currency from the bank's balance via \code{Coin::withdraw} and returns it to the caller.

\paragraph{Resources as Capabilities}
We conclude our discussion of the \move bank by noting that the advantage of an explicit type for money goes beyond safety: resources enable flexible programming patterns that would not be possible with an implicit representation of money.
For example: say that Alice is a customer of the \code{Bank} and wants to give another user Bob permission to withdraw the funds she has deposited.
Alice can simply transfer ownership of her \code{Bank::Credit} to Bob, who can use it to invoke \code{withdraw} at his leisure---no change to the \code{Bank} code is required.
Bob could also choose to store his \code{Bank::Credit} in another resource that (e.g.) allows multiple parties to access it or prevents it from being redeemed until a certain time.

By contrast, the Solidity and Scilla implementations of the \code{Bank} cannot support this feature without modifying the original contract to support it.
In essence, the \code{credit} map approach implements an access control list for withdrawing native currency, whereas the resource approach implements a linear \emph{capability} for withdrawals~\cite{Hardy94theconfused, capability-myths, DBLP:journals/pacmpl/Swasey0D17}.
Capability-based programming enables some powerful design patterns, as we will see in \Cref{sec:core_modules}.

\section{Move Overview}
\label{sec:move_overview}
This section provides an informal overview of the key concepts and design decisions of the \move language
that support safe and expressive programming with resources.

\subsection{Executable Bytecode With Resources}
\label{sec:bytecode}

The \move execution platform relies on a compiler to transform source language
programs into programs in the \move bytecode language. For example, \Cref{fig:bank} contains \move source code that compiles to an executable bytecode representation (see \Cref{fig:source_to_bytecode} for an example). Bytecode -- not source code -- is stored and executed on the \libra blockchain.

Because \move programs are deployed in the open alongside other (potentially untrusted) \move programs, it is important for key properties like resource safety to hold for \emph{all} \move bytecode programs. If the safety guarantees were only enforced by the source language compiler, an adversary could subvert them by writing malicious bytecode directly and entering it into the execution environment without using a compiler. Thus, we focus on the design and semantic properties of the \move bytecode language here, although we write illustrative examples in the source language for readability.

The \move execution platform relies on a load-time \emph{bytecode verifier},  in a manner similar to the Java Virtual Machine~\cite{jvm} and Common Language Runtime~\cite{clr}.
The bytecode verifier enforces type, memory, and resource safety.
Because the goal of the present paper is to explain and formalize properties of \move{}
that provide key advantages over prior smart contract languages (i.e. resource values with ironclad safety guarantees),
we focus on a concrete semantics for \move with dynamic checks for type, resource,
and memory safety and leave formalization of the bytecode verifier to future work.
Our formalization and resource safety theorem (\Cref{thm:resourcesafe}) therefore do not depend on any of the invariants ensured by the bytecode verifier; the presence of the verifier just allows an optimized implementation to skip these checks.
The analyses performed by the bytecode verifier are sufficiently interesting and complex to fill a paper of their own (particularly reference safety, which has similarities to the Rust borrow checker; see \Cref{sec:related_work}).

\paragraph{Persistent Global State}

\begin{figure}
\scriptsize
\begin{tabular}{lr}
\begin{minipage}{0.40\linewidth}
\begin{tabular}{ll}
Bytecode & Source code \\
\hline
$\ao{\movelocCmd}{x_0}$ & \smallcode{move credit} \\
$\ao{\unpackCmd}{s_1}$ & \smallcode{Credit \{amt, bank\}=...} \\
$\ao{\borrowglobalCmd}{s_0}$ & \smallcode{borrow_global<T>(...)} \\
$\ao{\borrowfieldCmd}{f_0}$ & \smallcode{&mut t.balance} \\
$\ao{\callCmd}{h_0}$ & \smallcode{Coin::withdraw(...)} \\
$\returnCmd$ & \smallcode{return} \\
\end{tabular}
\end{minipage}
&
\begin{minipage}{0.60\linewidth}
\includegraphics[scale=0.16]{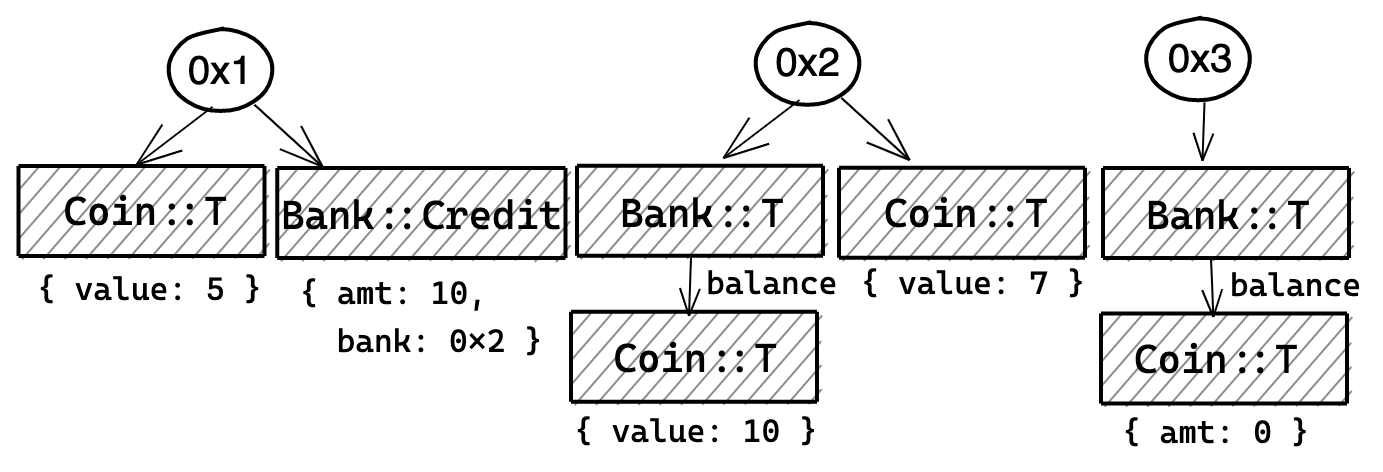}
\end{minipage}
\end{tabular}
\normalsize
\caption{Bytecode for the \code{withdraw} procedure from \Cref{fig:bank} (left) and an example global state containing resources from \Cref{fig:bank} (right). The global state contains three account addresses with different combinations of resources. Address \code{0x1} has a \code{Coin} resource with value 5 and a \code{Credit} with value 10 that can be redeemed at the \code{Bank} resource owned by \code{0x2}. Address \code{0x3} also has a \code{Bank} resource, but it holds a \code{Coin} with value 0.}
\label{fig:source_to_bytecode}
\end{figure}

\move execution occurs in the context of a persistent global state organized as a partial map from \emph{account addresses} to resource data values.
Each address can store an arbitrary number of resources, but at most one of any given type at the top level. For example, the account address \code{0x} in \Cref{fig:source_to_bytecode} holds two \code{Coin::T} resources, but one is at the top level and one is stored inside a \code{Bank::T} resource.

In addition, an address can store zero or more code \emph{modules}.
The global state is updated via \emph{transactions} that contain a sender account address and a \emph{transaction script} consisting of a single \code{main} procedure.
Transaction scripts update the global state by invoking procedures of published modules that mutate stored resources, add new resources to an address, or remove existing resources from an address.
A transaction has all-or-nothing semantics; either the entire script is executed without errors or it \emph{aborts} and reverts all changes to the global state.

\paragraph{Procedure Calls}
\begin{figure}
\begin{tikzpicture}
  \tikzstyle{frame}=[draw, fill=blue!20, text width=4em,
    text centered, minimum height=2.5em]
\tikzstyle{newFrame}=[draw, dashed, fill=blue!20, text width=4em,
    text centered, minimum height=2.5em]
\tikzstyle{stack} = [frame, text width=3.5em, fill=red!20,
    minimum height=12em, rounded corners]
\tikzstyle{glob} = [frame, text width=3.5em, fill=green!20,
    minimum height=12em]

\node (opstack) [stack] {Operand\\ stack};
\path (opstack) node (locals1) [frame, left=2cm of opstack] {Locals$_{i}$};
\path (opstack) node (locals0) [newFrame, above=0.5cm of locals1] {Locals$_{i+1}$};
\path (opstack) node (locals2) [frame, below=0.5cm of locals1] {\ldots};
\path (opstack) node (globals) [glob, right=2cm of opstack] {Global\\ resources};

\draw[->, thick, dotted] (opstack.west |- locals0.south) -- node [above] {\callCmd} (locals0.south east);
\draw[->, thick, dotted] (locals0.north east) -- node [above] {\returnCmd} (opstack.west |- locals0.north);

\draw[->, thick] (opstack.west |- locals1.south) -- node [below] {\storelocCmd} (locals1.south east);
\draw[->, thick] (locals1.north east) -- node [above] {\movelocCmd\ \copylocCmd} (opstack.west |- locals1.north);
\draw[->, thick] (locals1.north east) -- node [below] {\borrowlocCmd} (opstack.west |- locals1.north);

\draw[->, thick] ([yshift=0.5cm]opstack.east) -- node [above] {\movetoCmd} ([yshift=0.5cm]globals.west);
\draw[->, thick] ([yshift=-0.5cm]globals.west) -- node [above] {\movefromCmd} ([yshift=-0.5cm]opstack.east);
\draw[->, thick] ([yshift=-0.5cm]globals.west) -- node [below] {\borrowglobalCmd} ([yshift=-0.5cm]opstack.east);
\end{tikzpicture}
\caption{Execution mechanics of the \move bytecode interpreter. The global state holds resources that can be moved onto the operand stack or borrowed by pushing a reference onto the stack. Resources can be published to the global state by moving them from the stack into an account address. Each call stack frame (blue) has its own local variables to store values popped off the stack. Formal parameters and return values are passed between caller and callee using the shared operand stack.}

\label{fig:interpreter_mechanics}
\end{figure}
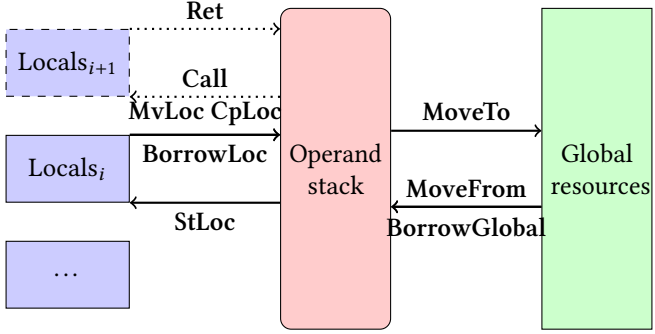

\begin{figure}
\begin{tabular}{|ll|}\hline
local variable instructions & \ao{\movelocCmd}{x} $|$ \ao{\copylocCmd}{c} $|$ \ao{\storelocCmd}{x} $|$ \ao{\borrowlocCmd}{x} \\
reference instructions & \readrefCmd $|$ \writerefCmd $|$ \freezerefCmd \\
record instructions & \packCmd $|$ \unpackCmd $|$ \ao{\borrowfieldCmd}{f} \\
global state instructions & \ao{\movetoCmd}{s} $|$ \ao{\movefromCmd}{s} $|$ \ao{\borrowglobalCmd}{s} $|$ \ao{\existsCmd}{s} \\
stack instructions & \popCmd $|$ \ao{\loadconstCmd}{a} $|$ \stackopCmd \\
procedure instructions & \ao{\callCmd}{h} $|$ \returnCmd  \\
\hline
\end{tabular}
\caption{List of \move instructions. The local variable instructions move or copy values between local variables and the operand stack. Reference instructions operate on reference values stored on the operand stack. The global state instructions move values between the operand stack and persistent global storage. Stack instructions manage the operand stack by popping unused values, pushing constants, and performing arithmetic/bitwise operations via \stackopCmd. Finally, the procedure instructions create and destroy call stack frames.}
\label{fig:instructions}
\end{figure}

Execution of a \move program begins by executing the distinguished \code{main} procedure of the transaction script and proceeds via the evaluation mechanics shown in \Cref{fig:interpreter_mechanics}.
A procedure is defined by a type signature and an executable body comprising a linear sequence of \move bytecode commands.
Procedure calls are implemented using a standard call stack containing frames
with a procedure name, a set of local variables, and a return address.
When one procedure calls another, the
calling procedure pushes its callee's arguments onto the operand stack and
invokes the \callCmd bytecode command, which pops the arguments off the stack and stores them in the actuals of the callee
(which become a subset of the callee's local variables).
Before returning, the callee pushes its return values on the stack and invokes
the \returnCmd bytecode command, which pops the current stack frame and returns control to the return address.

\paragraph{Modules}
A \move module such as our \code{Bank} from \Cref{fig:bank} can declare both record types and procedures.
Records can store primitive data values (including booleans, unsigned integers, and account addresses) as well as other record values, but not references.
Each record is nominally declared as a resource or non-resource.
Non-resource records cannot store resource records, and only resources can be stored in the global state.

Modules support strong encapsulation for their declared types.
Consider the bytecode translation of the \code{withdraw} procedure from our running example shown in \Cref{fig:source_to_bytecode}.
The struct definitions $s_i$ and field definitions $f_i$ used by the bytecode instructions
are implemented as integer indexes into internal tables of the current module.
This design ensures that privileged operations on the module's declared types can only be performed by procedures in the module,
encapsulating creation via \packCmd, destruction via \unpackCmd, accessing fields via \borrowfieldCmd, publishing via \movetoCmd, removing via \movefromCmd, and accessing (either to read or write) via \borrowglobalCmd.
For example: the \code{withdraw} bytecode is able to access a field of its declared \code{T} type, but it would not be able to access a field of the \code{Coin::T} type except via the API exposed by the \code{Coin} module.

A module may import a type or procedure declared in another module using its storing address as a namespace.
For example, the \code{use 0x0::Coin} line from our running example indicates that the current module should link against the module named \code{Coin} stored at account address \code{0x0}.
The combination of encapsulation and resource safety enables modules to safely interoperate while maintaining strong internal invariants.

\paragraph{References}
Move supports references to records and primitive values (but not to other references). In a manner similar to Rust, references are either exclusive/mutable (written \code{&mut}) or read-only (written \code{&}). All reads and writes of record fields occur through a reference.

References are different from other Move values because they are transient: as explained above, persistent global state consists of resource records, which cannot have fields of reference type. This means that each reference must be created during the execution of a transaction script and released before the end of that transaction script. Thus, each individual record value is a tree, and the global state is a forest whose roots are account addresses.

\subsection{Language Design for Resource Safety}
\label{sec:language_design}

\begin{figure}[t]
\centering
\begin{tabular}{cc}
\begin{minipage}{0.5\linewidth}
{
\begin{lstlisting}[language=move,basicstyle=\scriptsize\ttfamily]
resource R {}
fun copy_resource_bad(r: R) {
  let x = copy r; // no; copies r
}
fun deref_resource_bad(ref: &R): R {
  return *ref; // no; copies target of ref
}
\end{lstlisting}
}
\end{minipage}
&
\begin{minipage}{0.5\linewidth}
{
\begin{lstlisting}[language=move,basicstyle=\scriptsize\ttfamily]
resource R {}
fun double_move_bad(r: R): R {
  let x = move r;
  let y = move r; // no; r already moved
  let z = move_from<R>(0x1);
  return move_from<R>(0x1); // no; 0x1.R moved
}
\end{lstlisting}
}
\end{minipage}
\\[-5pt]\\\hline
\\[-5pt]
\begin{minipage}{0.5\linewidth}
{
\begin{lstlisting}[language=move,basicstyle=\scriptsize\ttfamily]
resource R {}
fun destroy_via_assign_bad(r1: R, r2: R) {
  let loc = move r1;
  loc = move r2; // no; destroys old value of loc
}
fun destroy_via_write_bad(ref: &mut R, r: R) {
  *ref = move r; // no; destroys target of ref
}
fun unused_resource_local_bad(r: R) {
  let local = move r;
  return; // no; would destroy resource in local
}
\end{lstlisting}
}
\end{minipage}
&
\begin{minipage}{0.5\linewidth}
{
\begin{lstlisting}[language=move,basicstyle=\scriptsize\ttfamily]
resource R {}
fun double_move_to_bad(r1: R, r2: R)) {
   move_to<R>(0x1, move r1);
   move_to<R>(0x1, move r2); // no; overwrites r1
}

\end{lstlisting}
}
\end{minipage}
\end{tabular}
\caption{Top: Examples of bad \move code that must be rejected due to resource duplication. Bottom: examples of bad \move code that must be rejected due to destruction of resources. The programs on the left would all be accepted if type \code{R} was declared as a \code{struct} instead of a \code{resource}, or if \code{R} was replaced with a primitive type like \code{u64}.}
\label{fig:bad_programs}
\end{figure}

At the beginning and end of a transaction script, all of the resources in the system reside in the global state $\globalstates$.
Resource safety is a conservation property that relates the set of resources present in state $\globalstates_{pre}$ before the script
to the set of resources present in state $\globalstates_{post}$ after the script.
In general terms, we would like the language to guarantee that:
\begin{enumerate}
\item A resource \code{M::T} that is present in post-state $\globalstates_{post}$ was also present in pre-state $\globalstates_{pre}$ unless it is introduced by a \packCmd inside \code{M} during script execution
\item A resource \code{M::T} that was present in pre-state $\globalstates_{pre}$ is also present in post-state $\globalstates_{post}$ unless it is eliminated by an \unpackCmd inside \code{M} during script execution
\end{enumerate}

It is helpful to look at each of the instructions in \Cref{fig:instructions} and consider what precautions must be taken in order to ensure that properties (1) and (2) hold. For property (1), we must be careful not to introduce instructions that can duplicate a resource value. \move achieves this by providing both \movelocCmd and \copylocCmd instructions for transferring a value from a local variable to the operand stack. As the \code{copy_resource_bad} function in \Cref{fig:bad_programs} demonstrates, the \copylocCmd instruction cannot be applied to a resource value. The \movelocCmd, \movetoCmd, and \movefromCmd instructions for transferring values prevent double moves that would allow a programmer to ``spend'' the same resource multiple times (see \code{double_move_bad}).

References must also be managed carefully to avoid duplication. The \readrefCmd for dereferencing a reference value can only be applied to a non-resource reference. Allowing a dereference of a resource like \code{deref_resource_bad} in \Cref{fig:bad_programs} would copy the resource value behind the reference.

Property (2) is further challenging because conventional languages provide a number of ways to indiscriminately discard values.
At the instruction level, restrictions must be placed on \popCmd, \storelocCmd, \writerefCmd, and \movetoCmd.
Most obviously, popping a resource off the operand stack with \popCmd must be disallowed.
If a local variable is of type resource, \storelocCmd can only be applied when the variable is uninitialized.
Code like \code{destroy_via_assign_bad} in \Cref{fig:bad_programs} would violate property (2) by discarding the old value stored in the local.
Similarly, a \writerefCmd like \code{*ref = move r} in \code{destroy_via_write_bad} must not execute if \code{ref} points to a resource.
This destructive update would destroy the value previously pointed to by \code{ref}.
Finally, the \movetoCmd instruction for moving a resource into global storage aborts if the move would overwrite an existing resource at the given address. For example, \code{double_move_to_bad} would fail at runtime because the memory \code{0x1.R} is already occupied.

Instruction-level protections are not quite enough to ensure property (2). There are two remaining holes that could allow resource destruction: values left in local variables when a procedure returns (e.g., \code{unused_resource_local_bad} in \Cref{fig:bad_programs}), and values left on the operand stack at the end of script execution. \move prevents both with extra discipline in the calling convention:
\begin{itemize}
\item The values on the operand stack match the types of formal parameters/return values before a \callCmd/\returnCmd (respectively).
\item \returnCmd cannot be invoked if a local variable holds a resource value or the operand stack holds extra (non-return) values.
\item A script terminates in a non-aborting state only when both the call stack and operand stack are empty.
\end{itemize}

The reader might wonder: can resources left on the stack and in locals be destroyed by a mid-script abort? This would be indeed be a problem in a conventional language, but the all-or-nothing semantics of \move transactions saves us. An aborting transaction script evaluates to the pre-state of the script, at which point all resources reside safely in global storage.

\paragraph{What Resource Safety Accomplishes for Programmers}
At this point, it's worthwhile to take a step back and briefly discuss what resource safety does and does not guarantee.
For concreteness, let's consider our running example in \Cref{fig:bank} once more.
Resource safety would \emph{not} preclude an implementation of \code{deposit} whose first line was \code{let amt = 7}; that is, it cannot protect the programmer from mistakes in implementing a custom asset.
It does, however, isolate and localize such decisions.
For example, it prevents the \code{Bank} from violating the invariants established for the imported \code{Coin::T} type inside its own declaring module.

This observation suggests a clear division of responsibilities.
It is the module author's job to define and correctly implement safety invariants for the types inside her module.
Once she has done so, encapsulaton and resource safety will ensure that her local invariants are also global invariants---no possible client can ever violate them (similar to the ``robust safety'' of \cite{DBLP:journals/pacmpl/Swasey0D17}).
This is quite powerful because \move modules give programmers an unusual amount of
control over declared types (e.g., restricting publishing and destroying types
as described above), and this control can be used to establish strong invariants.
For example, it is possible to define a type that can only be created after a certain time, a type that can never be destroyed, or a type that can only be created by a caller that has paid ten coins.
In \Cref{sec:core_modules}, we will show how resource safety allows us to establish global conservation of native currency in the \libra platform via a local invariant of the \code{Coin} module.

\section{\move Bytecode Interpreter}
\label{sec:bytecode_formal}

Next, we present operational semantics for a call-free subset of the \move bytecode that simulates a single transaction of arbitrary length. Generalizing to multiple transactions with procedure calls is conceptually straightforward, but would be significantly less concise.
This semantics will be used in \Cref{sec:resource_safety} to formalize and prove resource safety.

As explained in \Cref{sec:bytecode}, \move uses a bytecode verifier to ensure type safety and memory safety of smart contracts. Our formalism here, focusing on resource safety, does not depend on the bytecode verifier. Instead, our semantics gets stuck in errournous states, e.g. when encountering a dangling reference or an ill-typed operation.
The bytecode verifier ensures additional invariants (e.g., no dangling references, well-typedness) such that programs that pass the bytecode verifier cannot get stuck
due to memory or type errors.
As a result, our resource safety theorem (\Cref{thm:resourcesafe}) does not depend on the bytecode verifier.

We will begin with preliminary definitions and notation
for values, types, memory, and persistent global state, before introducing evaluation rules.
The notation is summarized in \Cref{tab:defs}.

\begin{figure*}
\begin{tabular}{|ll|}\hline
locations & \cl \\
primitive data values & \prival \\
addresses & $\accadd\subseteq \prival$  \\
resource types & \structname \\
resource tags & \rn \\
tags & $\tg=\rn\uplus \set{\unres}$\\
field names  & \fn \quad (finite)\\
paths & $\fn^{\ast}$ \\
values & \val \quad (see \Cref{def:valtagval})\\
tagged values & \tagval \quad (see \Cref{def:valtagval})\\
record values & $\val\setminus\prival$ \\
memories & $\mem=\cl\pfun\tagval$  \\
references & $\references=\cl\times \fn^{\ast} \times \set{\mut,\immut}$  \\
stack values & $\stackvalues=\tagval\cup\references$\\
local values & $\localvalues=\cl \cup\references$\\
local variables & \varid \\
local states & $\mem\times(\varid\pfun\localvalues)\times \stackvalues^{\ast}$ \\
global resource ids & $\globalresid=\accadd\times\structname$ \\
global states & $\globalstates=\mem\times(\globalresid\pfun\cl)\times(\varid\pfun\localvalues)\times \stackvalues^{\ast} $ \\
program locations & $PC$ \\
program states & $PC\times\globalstates$\\\hline

\end{tabular}
\caption{\label{tab:defs}Definitions for the semantics of \move. For a set $X$,
$X^{\ast}$ denotes the set of (finite) lists of elements from $X$.}
\end{figure*}

\subsection{Definitions and Notation}

\paragraph{Notation for partial functions and lists}
We use standard operations on partial functions (used to represent record values
or mappings in  local and global states); operations on lists are similarly standard
and used in several ways.

Following common convention,
if $f:A\pfun B$ is a partial function from $A$ to $B$, then $\dom{f}$ is the set of all $a\in A$
for which $f(a)$ is defined, and $\img{f}$ is the set of all $b\in B$ for which
$f(a)=b$ for some $a\in A$.  We use $\lset{f}{a}{b}$ to denote the function
that is equivalent to $f$ on every input except $a$ and which maps $a$ to $b$.
Similarly, $\mdel{f}{a}$ is the partial function equivalent to $f$ except that it is
undefined at $a$.

We use lists to represent a sequece of field accesses
and in components of semantic states.
We write $[]$ for the empty list and $e::l$ for the result of placing $e$ at the front of list $l$.
Similarly, $l::e$ is the list with $e$ appended to $l$ and, by slight abuse of notation,
$l::l'$ the concatenation of lists $l$ and $l'$.

\paragraph{Values and their types}
We begin with primitive types, field names, and tags, using these three elements to
define the values used in computation.
While tags are used to state and prove semantic properties,
tags are not needed in the \move virtual machine.

Let $\prival$ be the set of primitive data values, including Booleans, integers, and addresses,
$\fn$ a fixed, finite set of \emph{field names}
and  $\tg$ the set of \emph{tags}, where each tag may be a resource tag from a set $\rn$ or
the distinguished element $\unres$ indicating a value that is not a resource.

\begin{definition}
\label{def:valtagval}
The set $\val$ of \emph{values} and the set $\tagval$ of \emph{tagged values}
are defined together from the primitive values, tags  and field names
as the least sets satisfying:
\begin{enumerate}[(i)]
\item $\prival \subseteq \val$;
\item for every $v\in\val$ and $t\in\tg$, $\tv{v}{t}\in\tagval$; and
\item if $n\geq 1$, $f_1\dots f_n\in\fn$ are pair-wise distinct, and $tv_1\dots tv_n\in \tagval$, then
  $\{(f_i,tv_i)|1\le i \le n\} \in \val$.
\end{enumerate}
The values arising from condition (iii) are non-empty partial functions from $\fn$ to $\tagval$,
which we call \emph{record} values.
We use ordinary function notation when using them (e.g., when writing
$v(f)$ to refer to the value associated with field $f$ in the record value
$v$).
\end{definition}

Although types are not used extensively in this paper, we leverage
the fact that typing distinguishes resource values from non-resource values.
We write $\typeof{v}{s}$ to indicate that value $v$ has a \emph{type} $s$.
If $\structname$ is the set of resource types,
$\typeof{v}{s}$, and $s\in \structname$, then we say
that $v$ is a {\em resource value} and $\tv{v}{t}$ is a {\em resource tagged value},
or simply {\em resource};
otherwise, $v$ is a {\em non-resource value}
and $\tv{v}{t}$ is a non-resource tagged value.
\footnote{
In a well-formed tagged value, the type of the value must be consistent with
the tag (see \Cref{def:unrestricted}).
}

\paragraph{Paths and Trees}
In the semantics, a path is a possibly empty list of field names, which we think of as representing a sequence of field selections.
%

A tagged value may be regarded as a labeled tree, in the usual way that expressions are parsed as trees,
with nodes  labeled by tags and edges labeled by field names.
Specifically, a primitive value is a tree consisting of a leaf.
The tree associated with a tagged record value consists of a node labeled with the tag
and a subtree for each record component.
If $r$ is a record value, then for each $(f,tv)\in r$, there is an edge from $r$ to the subtree for $tv$ labeled by $f$.

Two useful operations on values and paths are
(i) the subterm $tv[p]$  of $tv$ located at path $p$, and
(ii) the term $tv[p:=tv']$ obtained by replacing the subterm
at path $p$ with term $tv'$.
The subterm
identified by following the empty path is the term itself,
i.e. $tv[[]] = tv$.
These operations are formalized as follows.

\begin{definition}

If $tv=\tv{v}{t}\in\tagval$, then $tv[p]$ is defined inductively by:
\begin{enumerate}
\item $tv[[]] = tv$
\item $tv[f::p'] = v(f)[p']$ if $v$ is a record value and $f\in\dom{v}$
\item undefined otherwise
\end{enumerate}
\noindent
Similarly, if $tv=\tv{v}{t}$ and $tv'$ are both tagged values, then $tv[p:=tv']$ is defined inductively by:
\begin{enumerate}
\item $tv[[]:=tv'] = tv'$
\item $tv[f::p':=tv'] = \tv{\lset{v}{f}{v(f)[p':=tv']}}{t}$ if $v$ is a record value
and $f\in\dom{v}$
\item undefined otherwise
\end{enumerate}
\end{definition}

\paragraph{States}
In the \move semantics, a state comprises
persistent global storage, local memory, operand stack and local variables.

If $\cl$ is the set of memory \emph{locations}, then a \emph{reference}
is a triple $\rft{c}{p}{q}$ consisting of a location $c\in\cl$,
path $p$, and mutability qualifier $q \in \{\mut, \immut\}$.

Local states and global states include memories, which are mappings
from locations to values, and stacks.
Specifically,
a \emph{memory} $M$ is a  partial function from $\cl$ to $\tagval$.
Defining \emph{local values} to be locations or references,
a \emph{local memory} is similarly a partial function from $\varid$ to local values,
where $\varid$ is a set of local variables.
A {\em local stack} is a list of \emph{stack values}, which may be tagged values or references.

Global resources are identified by an address and a resource type.
If $\accadd$ is the set of addresses, then
the set $\globalresid= \accadd\times \structname$ of \emph{global resource id}s consists of
pairs $\tup{a,R}$, each associating a primitive value $a\in \accadd$ of type addresses
with a resource type $R\in\structname$.
A \emph{global store} $G$ is a  partial function from global resource ids $\globalresid$ to $\cl$.

A \emph{global state} is a tuple $\st{M,G,L,S}$, where $M$ is a memory,
$L$ is a local memory,  $S$ is a local stack, and $G$ is a global store.
A  \emph{local state} is similar with the global store omitted.

A \move program $P$ is a mapping from program locations $PC$ to operations and,
if  $\pc \in PC$ represents the current program counter, then $P[\pc]$ is the current instruction
and $P[\pc+1]$ is the next instruction under normal execution.
A \emph{program state} consists of a program counter $\pc \in PC$ and a global state.

\begin{figure*}
\begin{centerframe}
\semrule[\movelocRule]
  {
               \lread{L}{x} = c\sep c\in\dom{M}
  }
  {\st{M,L,S}}
  {\ao{\movelocCmd}{x}}
  {\st{\mdel{M}{c},\mdel{L}{x},\stackht{\mread{M}{c}}{S}}}\\
\bigskip

\semrule[\movelocRefRule]
  {
               \lread{L}{x} = \rft{c}{p}{q}
  }
  {\st{M,L,S}}
  {\ao{\movelocCmd}{x}}
  {\st{M,\mdel{L}{x},\stackht{\lread{L}{x}}{S}}}\\
  \bigskip

\hspace*{-.7in}
\semrule[\hspace{-.5em}\copylocvalRule]
  {
    \lread{L}{x}=c\sep c\in\dom{M}\sep \mread{M}{c}=\tv{v}{\unres}
  }
  {\st{M,L,S}}
  {\ao{\copylocCmd}{x}}
  {\st{M,L,\stackht{\mread{M}{c}}{S}}}
\hspace*{.3in}
\semrule[\copylocrefRule]
  {
    r=\lread{L}{x} = \rft{c}{p}{q}
  }
  {\st{M,L,S}}
  {\ao{\copylocCmd}{x}}
  {\st{M,L,\stackht{r}{S}}}\\
  \medskip

\semrule[\cmdname{StLoc-TV}]
  {
    r\in\tagval\sep (x\not\in\dom{L} \vee \lread{L}{x} \in \references
    \vee     (\lread{L}{x} = c\wedge
    \mread{M}{c}=\tv{v}{\unres}))
    \sep c'\notin\dom{M}
  }
  {\st{M,L,\stackht{r}{S}}}
  {\ao{\storelocCmd}{x}}
  {\st{\mset{M}{c'}{r},\lset{L}{x}{c'},S}}\\
  \medskip

\semrule[\cmdname{StLoc-Ref}]
  {
    r\in\references\sep (x\not\in\dom{L} \vee \lread{L}{x}\in\references
    \vee     (\lread{L}{x} = c\wedge
    \mread{M}{c}=\tv{v}{\unres}))
  }
  {\st{M,L,\stackht{r}{S}}}
  {\ao{\storelocCmd}{x}}
  {\st{M,\lset{L}{x}{r},S}}\\
  \medskip



\semrule[\borrowlocRule]
  {
    \lread{L}{x}=c
  }
  {\st{M,L,S}}
  {\ao{\borrowlocCmd}{x}}
  {\st{M,L,\stackht{\rft{c}{[]}{\mut}}{S}}}\\
  \bigskip

\semrule[\borrowfieldRule]
  {
    r = \rft{c}{p}{q} \sep c\in\dom{M} \sep \mread{M}{c}[p]=\tv{\record{(f,tv_{f}),\cdots}}{t}
  }
  {\st{M,L,\stackht{r}{S}}}
  {\ao{\borrowfieldCmd}{f}}
  {\st{M,L,\stackht{\rft{c}{p::f}{q}}{S}}}\\
  \bigskip

\semrule[\freezerefRule]
  {
  }
  {\st{M,L,\stackht{\rft{c}{p}{q}
}{S}}}
  {\freezerefCmd}
  {\st{M,L,\stackht{\rft{c}{p}{\immut}}{S}}}\\
  \bigskip

\semrule[\readrefRule]
  {
    r = \rft{c}{p}{q}\sep
    c \in\dom{M}\sep
    \mread{M}{c}[p]=\tv{v}{\unres}
  }
  {\st{M,L,\stackht{r}{S}}}
  {\readrefCmd}
  {\st{M,L,\stackht{\tv{v}{\unres}}{S}}}\\
  \bigskip

\semrule[\writerefRule]
  {
    r = \rft{c}{p}{\mut}\sep
    tv=\mread{M}{c}\sep
    tv[p]=\tv{v}{\unres}\sep
    tv'=\tv{v'}{\unres} \sep
  }
  {\st{M,L,\stackht{tv'}{\stackht{r}{S}}}}
  {\writerefCmd}
  {\st{\mset{M}{c}{tv[p := tv']},L,S}}\\
%
  \bigskip

\begin{tabular}{ll}
\semrule[\popValRule]
  {
    tv=\tv{v}{\unres}
  }
  {\st{M,L,\stackht{tv}{S}}}
  {\popCmd}
  {\st{M,L,S}} &
\hspace*{.5in}
\semrule[\popRefRule]
  {
    r = \rft{c}{p}{q}
  }
  {\st{M,L,\stackht{r}{S}}}
  {\popCmd}
  {\st{M,L,S}}
\end{tabular}\\
\bigskip

\semrule[\packResourceRule]
  {
    s\in\structname\sep
    \typeof{\record{(f_{i},tv_{i})\mid 1 \leq i \leq n}}{s}
    \sep tv =
    \tv{\record{(f_{i},tv_{i})\mid 1 \leq i \leq n}}{t}
    \sep
    \mbox{$t\in\rn$ \text{ is fresh}}
  }
  {\st{M,L,\stackht{tv_{1}\cons\cdots\cons tv_{n}}{S}}}
  {\ao{\packCmd}{s}}
  {\st{M,L,
    \stackht{tv}{S}}
  }\\
  \bigskip

\semrule[\packUnrestrictedRule]
  {
    s\not\in\structname\sep
    \typeof{\record{(f_{i},tv_{i})\mid 1 \leq i \leq n}}{s}
    \sep
    tv =
    \tv{\record{(f_{i},tv_{i})\mid 1 \leq i \leq n}}{\unres}
    \sep
    tv_{i}=\tv{v_{i}}{\unres}
  }
  {\st{M,L,\stackht{tv_{1}\cons\cdots\cons tv_{n}}{S}}}
  {\ao{\packCmd}{s}}
  {\st{M,L,\stackht{tv}{S}}
  }\\
  \medskip

\semrule[\unpackRule]
  { tv=
    \tv{\record{(f_{i},tv_{i})\mid 1 \leq i \leq n}}{t}
  }
  {\st{M,L,\stackht{tv}{S}}}
  {\ao{\unpackCmd}{s}}
  {\st{M,L,
    \stackht{tv_{1}\cons\cdots\cons tv_{m}}{S}}
  }\\
  \medskip

\semrule[\loadconstRule]
  {
  }
  {\st{M,L,S}}
  {\ao{\loadconstCmd}{a}}
  {\st{{M},L,\stackht{\tv{a}{\unres}}{S}}}\\
  \medskip
%
%
\semrule[\stackopRule]
  {
    tv_{i}=\tv{v_{i}}{\unres}\sep\legal(\stackop,v_1,\dots,v_n)
  }
  {\st{M,L,\stackht{tv_{1}}{\cdots\stackht{tv_{n}}{S}}}}
  {\stackopCmd}
  {\st{M,L,\stackht{{\tv{\stackop(v_1,\dots,v_n)}{\unres}}}{S}}}\\

  \caption{\label{fig:rules-tree-1}Operational Semantics of \move: operations
    on local state}
\end{centerframe}
\end{figure*}

\subsection{Local State Rules}

Each rule in \Cref{fig:rules-tree-1} operates on local states (global storage
is unchanged and thus omitted to keep the presentation simple) and
takes the form
$$
\semrule[Rule1]
{\varphi}
{\st{M,L,S}}
{\ao{op}{\cdot}}
{\st{M',L',S'}}
$$
where Rule1 is the name of the rule,
$\varphi$ is a precondition for applying it,
$\st{M,L,S}$ and
$\st{M',L',S'}$
are local states,
and \ao{op}{\cdot} is an operation parameterized by a field, variable, or record declaration of the current module.
When there are no parameters, we simply write $\op$.
We use the following variable conventions: $c\in\cl$;
$x\in\varid$; $t\in\tg$; $v\in\val$; $tv\in\tagval$; $f\in\fn$; $p$ is a path;
$q$ is a mutability qualifier; $r$ is a stack value; and $s$ is a type.
%

The $\movelocRule$ rules show how the state changes when a local value is moved
from a local variable $x$ onto the stack.  Note that if the value moved is not
a reference, it is removed from memory when it is placed on the stack.  The
$\copylocvalRule$ rules copy local values to the stack.  In this case, the local
variable $x$ (and memory if applicable) retain their values.  Note that
these rules can only be applied if the local value is not a resource.
The $\storelocbotrefRule$ rules take the top stack value and store it
in the local variable $x$.  There are two versions of the rule, depending on
the current local value of $x$.  If $x$ has no value or contains a reference, it
is always possible to store the stack value in $x$
(note that we can always choose a $c'$
not currently in the domain of $M$).  However, if $x$ contains a tagged value,
then the rule can
only be applied if the tagged value is not a resource.

$\borrowlocRule$ pushes a reference to the local value in $x$ onto the stack.
$\borrowfieldRule$ takes a reference $r$ from the top of the stack and pushes a new
reference onto the stack that points to the tagged value in field $f$ of the
record pointed to by $r$.
$\freezerefRule$ turns a mutable reference into an immutable
reference.  $\readrefRule$ makes a copy of the tagged value pointed to by a reference on
top of the stack and pushes it onto the stack (note that the value must be
a non-resource). \readrefCmd can be applied to either a mutable reference or an immutable reference.
$\writerefRule$ takes a non-resource tagged value $tv'$
and a
reference $r$ from the stack, and replaces the tagged value $tv$
pointed to by $r$ (which must also be a non-resource and of the same type) by
$tv'$. It can only be applied when $r$ is a mutable reference.
The distinction between mutable and immutable references is not particularly
useful in the semantics. However, we include these qualifiers because they
are crucially important for the correct operation of the \move bytecode
verifier and because we want the semantics to accurately reflect the
real implementaion.

The $\popValRule$ rules pop a stack value off the top of the stack (so long as it is not
a resource).
The $\packCmd$ rules (\packResourceRule and \packUnrestrictedRule) create a record of a given type $s$.
The $\unpackRule$ rule decomposes a record into its fields.
 For resources, \packResourceRule pairs the record with a fresh resource tag,
i.e., in each application of the rule a new unique tag is created.
 The \unpackRule rule discards the tag associated with the unpacked record,
but freshness guarantees the discarded tag will not be reused.
  $\loadconstRule$
places a constant primitive value $a$ onto the stack.  $\stackopRule$ is a
meta-rule: there is an instance of the rule for every operation on primitive
data values (e.g. negation and conjunction on Booleans, addition and
subtraction on integers, etc.).  The instantiated rules are formed by replacing
$\stackopCmd$ by the specific operation and replacing $\legal$ by a condition
that specifies legal operands for each $\stackop$ (e.g. that the divisor is
non-zero for a division operation).

\subsection{Global State Rules}
The rules of \Cref{fig:opsemAcc} are similar except that they operate on global
states.
$\movetoCmd$ takes an address and a resource from
the stack and
puts the resource in the global storage at the location indexed by the address and
the resource type.
Conversely, $\movefromCmd$ removes a resource
from global storage and puts it on the stack.  $\borrowglobalCmd$ gets a
reference to a resource in global storage.  And finally,
$\existsCmd$ checks
whether the global storage currently contains a resource for a particular
global resource id.
In this rule, $b$ is set to true
if $\tup{a,s}$ is in the domain of $G$,
and is false otherwise.

\begin{figure*}
\begin{centerframe}
\begin{small}





\semrule[\movetoRule]
  {
    tv_{1} = \tv{a}{\unres}\sep
    a\in\accadd\sep
    tv_2 = \tv{v}{t}\sep
    \typeof{v}{s}\sep
    s\in\structname\sep
    \tup{a,s} \notin \dom{G}\sep
    c\notin\dom{M}
  }
  {\st{M,G,L,\stackht{\stackht{tv_{1}}{tv_{2}}}{S}}}
  {\ao{\movetoCmd}{s}}
  {\st{\mset{M}{c}{tv_2},\gset{G}{\tup{a,s}}{c},L,S}}
  \bigskip


\semrule[\movefromRule]
  {
    s\in\structname\sep
    tv = \tv{a}{\unres}\sep
    a\in\accadd\sep
    \gread{G}{\tup{a,s}} = c\sep
    \mread{M}{c}=tv'
  }
  {\st{M,G,L,\stackht{tv}{S}}}
  {\ao{\movefromCmd}{s}}
  {\st{\mdel{M}{c},\gdel{G}{\tup{a,s}},L,\stackht{tv'}{S}}}\\
  \bigskip


\semrule[\borrowglobalRule]
  {
    s\in\structname\sep
    tv = \tv{a}{\unres}\sep
    a\in\accadd\sep
    \gread{G}{\tup{a,s}} = c
  }
  {\st{M,G,L,\stackht{tv}{S}}}
  {\ao{\borrowglobalCmd}{s}}
  {\st{M,G,L,\stackht{\rft{c}{[]}{\mut}}{S}}}\\
  \bigskip


\semrule[\existsRule]
  {
    tv = \tv{a}{\unres}\sep
    a\in\accadd\sep
    b \Leftrightarrow \tup{a,s} \in \dom{G}
  }
  {\st{M,G,L,\stackht{tv}{S}}}
  {\ao{\existsCmd}{s}}
  {\st{M,G,L,\stackht{\tv{b}{\unres}}{S}}}\\

\caption{\label{fig:opsemAcc}Operational Semantics of \move: Global state operations}
\end{small}
\end{centerframe}
\end{figure*}

\subsection{Program State Rules}
The rules of \Cref{fig:pcrules} lift the small-step semantics presented so far
to semantics of call-free \move programs that can model loops and conditional branching using unstructured control-flow.
They assume an abstract set $PC$ of program locations
over which a program counter $\pc$ ranges.
A program $P$ is a mapping from such program locations to operations.  The combined global, local, and memory state (represented in the rule $\steppcRule$ as $\sigma$) is extended with the $\pc$
to obtain a \emph{program state}, and the rules simply implement sequential and
branching control flow in a straightforward way.

\begin{figure*}
\begin{centerframe}
\pcrule[\steppcRule]
  {
    \progloc{P}{\pc}=\op,
    \step{\sigma}{\op}{\sigma'}
  }
	{P}{\pcstate{\pc}{\sigma}}{\pcstate{\pc+1}{\sigma'}}
%
%
%
\\
\pcrule[\branchtruepcRule]
  {
	  \progloc{P}{\pc}=\ao{\branchCmd}{\ell},
    \mread{M}{c}=\tv{\true}{\unres}
  }
	{P}{\pcstate{\pc}{\st{M,L,\stackht{c}{S}}}}{\pcstate{\ell}{\st{M,L,S}}}\\
\pcrule[\branchfalsepcRule]
  {
	  \progloc{P}{\pc}=\ao{\branchCmd}{\ell},
    \mread{M}{c}=\tv{\false}{\unres}
  }
	{P}{\pcstate{\pc}{\st{M,L,\stackht{c}{S}}}}{\pcstate{\pc+1}{\st{M,L,S}}}


\caption{\label{fig:pcrules}Program counter rules}
\end{centerframe}
\end{figure*}

These evaluation rules intentionally get stuck in the presence of resource, type, or memory errors (e.g., $\copylocCmd$ on a variable that contains a resource).
As we mentioned in \Cref{sec:bytecode}, the \move bytecode verifier performs checks that preclude these errors.
However, there are two kinds of runtime errors not caught by the bytecode verifier that we also model as stuck execution for convenience:
\begin{enumerate}
\item Errors in $\stackopCmd$ such as division by zero and arithmetic
  over/underflow (which \move chooses to treat as errors);
\item Overwriting an existing global resource id in $\movetoCmd$ or accessing a global resource id that  does not exist in $\movefromCmd$ or $\borrowglobalCmd$.
\end{enumerate}
In practice, these runtime errors trigger an abort that terminates the current transaction and reverts any changes to global state.

\section{Resource Safety}
\label{sec:resource_safety}
In this section, we prove that the operational semantics introduced above enforces a conservation property: a resource cannot be created or destroyed except by the privileged \packCmd and \unpackCmd constructs available in its declaring module.
We define a set of {\em well-formed} states (\Cref{def:wellformedstate}),
show that the semantic rules preserve well-formedness (\Cref{prop:well-formed-preserved}),
and finally, that well-formedness guarantees resource safety (\Cref{thm:resourcesafe}).
We start by introducing the parts of a well-formed state.

\begin{definition}[Well-formed tagged value]
\label{def:unrestricted}
A tagged value $tv=\tv{v}{t}$ with $\typeof{v}{s}$ is {\em well-formed} if
$s\in\structname$ iff
$t\in \rn$, and in addition, one
of the following holds:
(i)~$v$ is primitive;
(ii)~$v=\record{(f_{i}, tv_{i}) \mid 1\leq i\leq n}$ such that every $tv_{i}$
is well-formed,
and if $s\notin\structname$ then
 for every $i$,
$tv_{i}$ is not a resource.
\end{definition}

\noindent
Intuitively, in a well-formed tagged value, a
resource value is never nested inside of an non-resource value,
and the tag corresponds to the type.

\begin{definition}[Globally consistent state]
We say that a state
$\st{M,G,L,S}$ is {\em globally consistent} if the following holds:
\begin{inparaenum}[(i)]
\item\label{it:values} Every tagged value in $\img{M}$ or in $S$ is well-formed;
\item\label{it:memory} $\dom{M}=\img{G}\cup (\img{L}\cap\cl)$;
\item\label{it:globaltypes} For every $\tup{a,s} \in \dom{G}$,
$\mread{M}{\gread{G}{\tup{a,s}}}=\tv{v}{t}$ with
 $\typeof{v}{s}$.
\end{inparaenum}
\end{definition}

\noindent
Intuitively,
(\ref{it:values}) means that tagged values in the state are well-formed;
(\ref{it:memory}) means that global resource ids and local variables only point to locations in the memory
(no dangling references)
and
the memory only contains locations pointed to by some global resource id or local variable
(no garbage);
(\ref{it:globaltypes}) means that global values have their expected types.

\begin{definition}[Tag-consistent state]
A state
$\st{M,G,L,S}$ is {\em tag-consistent} if the following holds:
\begin{inparaenum}[(i)]
\item\label{it:tag1}
if
$\mread{M}{c_1}[p_{1}]=\tv{v_{1}}{t}$,
$\mread{M}{c_2}[p_{2}]=\tv{v_{2}}{t}$, and
$t \neq \unres$,
then
$c_{1}=c_{2}$ and $p_{1}=p_{2}$;
\item\label{it:tag2}
If
$S=[\ldots,s_{i_{1}},\ldots,s_{i_{2}},\ldots]$,
$s_{i_{1}}[p_{1}]=\tv{v_{1}}{t}$,
$s_{i_{2}}[p_{2}]=\tv{v_{2}}{t}$, and
$t \neq \unres$,
then
$i_{1}=i_{2}$ and $p_{1}=p_{2}$;
and
\item\label{it:tag3}
It is never the case that
$S=[\ldots,s_{i},\ldots]$,
$M[c][p_{c}]=\tv{v_{1}}{t}$,
$s_{i}[p_{i}]=\tv{v_{2}}{t}$, and
$t \neq \unres$.
\end{inparaenum}
\end{definition}

\noindent
Intuitively, being tag-consistent means that resource tags are unique, i.e.
a resource tag can appear in the memory and the stack at most once.

\begin{definition}[Non-aliasing]
A state
$\st{M,G,L,S}$ is {\em non-aliasing} if the following holds:
\begin{inparaenum}[(i)]
\item\label{it:tree1} If $x_{1},x_{2}\in\dom{L}$ with $x_1\neq x_2$ and $\lread{L}{x_{1}},\lread{L}{x_{2}}\in\cl$,
then $\lread{L}{x_{1}}\neq\lread{L}{x_{2}}$;
\item\label{it:tree2} If $g_{1}, g_{2}\in\globalresid$ with $g_1 \neq g_2$,
then $\lread{G}{g_{1}}\neq\lread{G}{g_{2}}$; and
\item\label{it:tree3} If $g\in\globalresid$ and $x\in\dom{L}$ then
$\lread{G}{g}\neq\lread{L}{x}$.
\end{inparaenum}
\end{definition}

\noindent
Intuitively, a state is non-aliasing if different
global or local identifiers
cannot point to the same
memory location.

\begin{definition}[Well-formed state]
\label{def:wellformedstate}
A state
$\st{M,G,L,S}$ is {\em well-formed} if
it is globally consistent, tag-consistent, and non-aliasing.
\end{definition} %

\noindent
Well-formed states ensure that
global resource identifiers and local variables only point to locations that are in the memory,
and do not alias.
Note, however, that according to these semantics, a well-formed state may still
contain dangling {\em references}
i.e.,
$\rft{c}{p}{q}\in\img{L}\cup S$
 s.t. $c \not\in\dom{M}$, as well as aliasing between references.
As explained in \Cref{sec:bytecode}, the bytecode veirifer ensures stronger guarantees (e.g., no dangling references), but in this section we do not depend on these stronger invariants.

We now show that the operational semantics preserves well-formedness of states.

\begin{definition}[Well-formed execution sequence]
\begin{sloppypar}
Let $P$ be a program.
An {\em execution sequence} of $P$ is
$\pi = \pcstate{pc_0}{\sigma_0},\ldots,\pcstate{pc_n}{\sigma_n}$
such that $pc_i\in \dom{P}$ for every $0 \leq i \leq n$
and $\pcstep{P}{\pcstate{\pc_{i}}{\sigma_i}}{\pcstate{\pc_{i+1}}{\sigma_{i+1}}}$
for every $0\leq i < n$.
An execution sequence is called {\em well-formed} if each $\sigma_i$ is well-formed.
\end{sloppypar}
\end{definition}

\begin{proposition}
\label{prop:well-formed-preserved}
\begin{sloppypar}
Let $P$ be a program and $\pi = \pcstate{pc_0}{\sigma_0},\ldots,\pcstate{pc_n}{\sigma_n}$
an execution sequence of $P$.
If $\sigma_0$ is well-formed, then $\pi$ is well-formed, i.e., $\sigma_1,\ldots,\sigma_n$ are all well-formed.
\end{sloppypar}
\end{proposition}

\begin{proof}
The proof is by induction on $n$, and
amounts to a routine check that
the rules of \Cref{fig:rules-tree-1} and \Cref{fig:opsemAcc},
as well as
the \branchfalsepcRule and \branchtruepcRule rules of \Cref{fig:pcrules} preserve well-formedness.
We explicitly prove this for \movelocRule.
The rest are verified similarly.
{\bf Global Consistency:}
(\ref{it:values}) follows from the induction hypothesis
since the set of tagged values in the memory and the stack are not changed.
(\ref{it:memory})
is also preserved:
by the induction hypothesis, $\tup{M,G,L,S}$ is non-aliasing.
Thus, the fact that $x$ is removed from $L$ also means that $c$
is removed from $\img{G}\cup (\img{L}\cap\cl)$.  Since $c$ is also removed
from $M$, it follows that
(\ref{it:memory}) holds.
(\ref{it:globaltypes}) is preserved since
$G$ is unchanged and no locations are added to $M$.
{\bf Tag Consistency:}
(\ref{it:tag1}) is preserved
as the memory only gets smaller after \movelocCmd.
(\ref{it:tag2}) and (\ref{it:tag3}) hold initially by the induction hypotehsis;
it is easy to see that both must also hold after moving a value from memory to
the stack.
{\bf Non-Aliasing:}
(\ref{it:tree2}) holds since the global state is unchanged.
Additionally, (\ref{it:tree1}) and (\ref{it:tree3}) are preserved
as $L$ only gets smaller after \movelocCmd.
\end{proof}

Next, we define the {\em resources} of a state,
and what it means for resources to be {\em introduced} or {\em eliminated} in an execution
sequence.  We can then prove the resource safety theorem.

\begin{definition}[State Resources]
Let $\sigma=\st{M,G,L,S}$ be a state.
The {\em resources of $\sigma$}, denoted
$\resourcesof{\sigma}$,
are defined as follows:
$\resourcesof{\st{M,G,L,S}}=
\set{t\in\rn \mid \tv{v}{t}\in\img{M}\cup S
}
$
\end{definition}

\noindent
Intuitively, resources of a state are the resource tags
that occur in a tagged value of the state.

\sloppy
\begin{definition}[Resources Introduced and Eliminated]
  \label{def:rie}
Let $P$ be a program and $\pi = \pcstate{pc_0}{\sigma_0},\ldots,\pcstate{pc_n}{\sigma_n}$
an execution sequence of $P$.
The set of {\em resources introduced} in $\pi$, denoted $\resourcesi{\pi}$, is:
$\{t \in \rn \mid \exists\, 0\leq i < n.\: \progloc{P}{\pc_{i}}=\packCmd$ and $\sigma_{i+1}=\st{M,G,L,\stackht{\tv{v}{t}}{S}}  \}$.
The set of {\em resources eliminated} in $\pi$, denoted $\resourcese{\pi}$, is:
$\{t \in \rn \mid \exists\, 0\leq i < n.\: \progloc{P}{\pc_{i}}=\unpackCmd$ and $\sigma_{i}=\st{M,G,L,\stackht{\tv{v}{t}}{S}}  \}$.
\end{definition}

\noindent
Intuitively, $\resourcesi{\pi}$ collects all resource tags that were
created (using $\packCmd$) during the execution; similarly,
$\resourcese{\pi}$ collects all resource tags that were consumed
(using $\unpackCmd$) during the execution. Notice that these sets are
not necessarily disjoint. That is, a resource that is created and later
consumed during $\pi$ will appear both in $\resourcesi{\pi}$ and in
$\resourcese{\pi}$.

\begin{theorem}[Resource Safety]
\label{thm:resourcesafe}
Let $P$ be a program and $\pi = \pcstate{pc_0}{\sigma_0},\ldots,\pcstate{pc_n}{\sigma_n}$
a well-formed execution sequence of $P$.
Then, $\resourcesof{\sigma_{n}} =
\resourcesof{\sigma_{0}}\cup \resourcesi{\pi} \setminus \resourcese{\pi}$.
\end{theorem}

\begin{proof}
  The proof is by induction on $n$.  The base case ($n=0$) is
straightforward (in this case, $\resourcesi{\pi}  = \resourcese{\pi} =
\emptyset$).  For the induction step, the induction hypothesis
provides:
$$(*)~\resourcesof{\sigma_{n-1}} =
\resourcesof{\sigma_{0}}\cup \resourcesi{ \pi' } \setminus
\resourcese{ \pi'}$$
where
$\pi' =
\pcstate{pc_0}{\sigma_0},\ldots,\pcstate{pc_{n-1}}{\sigma_{n-1}}$.
If
$\progloc{P}{\pc_{n-1}}\notin\set{\ao{\packCmd}{s},\unpackCmd}$, then
examination of the rules shows that $\resourcesof{\sigma_{n}} =
\resourcesof{\sigma_{n-1}}$ (i.e. for all rules other than the $\packCmd$ and
$\unpackCmd$ rules, the set of resource tags
in the global state remains the same after the application of the rule).
 By \Cref{def:rie},
$\resourcesi{\pi} = \resourcesi{\pi'}$ and $\resourcese{\pi} =
\resourcese{\pi'}$.
Using $(*)$,
we get
$\resourcesof{\sigma_{n}} =
  \resourcesof{\sigma_{0}}\cup \resourcesi{\pi} \setminus \resourcese{\pi}$.
  The proof is similar if $\progloc{P}{\pc_{n-1}}=\ao{\packCmd}{s}$ and $s\notin\structname$
  or if $\progloc{P}{\pc_{n-1}}=\unpackCmd$ and $\sigma_{n-1}=\st{M,G,L,\stackht{\tv{v}{\unres}}{S}}$.
  If $\progloc{P}{\pc_{n-1}}=\ao{\packCmd}{s}$ for a resource type $s$, then
  by $\packResourceRule$, $\resourcesof{\sigma_{n}} = \resourcesof{\sigma_{n-1}} \cup \set{t}$
  where $\sigma_n = \st{M,G,L,\stackht{\tv{v}{t}}{S}}$. We know $t$ is
  fresh, so it is not in $\resourcese{\pi'}$. By \Cref{def:rie}, $\resourcesi{\pi}\! =\! \resourcesi{\pi'} \cup \set{t}$ and $\resourcese{\pi} = \resourcese{\pi'}$. Thus, by $(*)$, $\resourcesof{\sigma_{n}} =
  \resourcesof{\sigma_{0}}\cup \resourcesi{\pi} \setminus \resourcese{\pi}$.
The proof is similar if
$\progloc{P\!}{\pc_{n-1}}\!=\!\unpackCmd$
  and $\sigma_{n-1}\!=\!\st{M,G,L,\stackht{\tv{v}{t}}{S}}$, $t\!\neq\!\unres$.
  %
\end{proof}

\section{Experience With \move}
\label{sec:experience}

In this section, we describe the open-source implementation of the \move language, report on our experience using \move in the \libra blockchain, and  mention efforts that have adopted or built on the language.

\subsection{Implementation}

\paragraph{Move Compiler}
We have implemented\footnote{\url{https://github.com/libra/libra/tree/master/language/move-lang}} a compiler from the \move source code used in \Cref{fig:bank} and \Cref{fig:bad_programs} to the \move bytecode language. The source language adds structured control flow for convenience and expressions to abstract away the operand stack, but the programming model otherwise matches the bytecode language. Although all of the examples in this paper use explicit \code{copy} and \code{move} directives when accessing variables (e.g., \code{let x = copy r}), the compiler does not require these directives. In the absence of a directive, the compiler uses liveness analysis to emit a \code{move} for the last usage of a variable and a \code{copy} for all other uses. In addition, the compiler implements source code equivalents of the bytecode verifier analyses with friendly error messages.

\paragraph{\move Virtual Machine}
The \move virtual machine implements a superset of the bytecode interpreter semantics described in \Cref{sec:bytecode_formal}. The implemented interpreter includes gas metering similar to the EVM~\cite{ethereum}, support for a limited form of generics, and a \code{vector} type.
The virtual machine also includes the bytecode verifier, which performs static checks for type safety, usage of uninitialized variables, reference safety, and stack balancing (to ensure that the callee cannot illegally access stack locations belonging to a caller).
The bytecode verifier has a linker for ensuring that the usage of external types in a module are consistent with their declarations (e.g., procedure \code{p} invoked in module \code{M1} exists and matches the type signature of its declaring module \code{M2}).
The implementation\footnote{\url{https://github.com/libra/libra/tree/master/language}} of both components consists of about 17K lines of Rust code.

Although the \move language was originally created to serve as the execution layer for the \libra blockchain, we have maintained a clean separation between the platform-agnostic \move language layer implemented in the virtual machine and the \libra-specific layer implemented in the \libra adapter component and \libra's \move standard library.
This flexible architecture has facilitated adoption of \move outside of \libra (see \Cref{sec:move_usage}).

\paragraph{Tooling}
In addition to the compiler and virtual machine, we have implemented several tools\footnote{\url{https://github.com/libra/libra/tree/master/language/tools}} to facilitate testing and analysis of \move code:
\begin{itemize}
\item A testing framework\footnote{\url{https://libra.org/en-US/blog/how-to-use-the-end-to-end-tests-framework-in-move/}} that allows users to write multi-transaction test scenarios.
\item A bytecode code coverage tool that attaches to the testing framework and records the bytecode instructions exercised by each test.
\item A \move bytecode disassembler similar to the \texttt{javap} utility for Java bytecode. The disassembler prints raw bytecode, but can also accept an optional source to bytecode map that augments the result with variable names and line numbers.
\end{itemize}

\subsection{Integration With the \libra Blockchain}
\label{sec:core_modules}
The \move VM implements the transaction execution layer in the \libra blockchain~\cite{libra_blockchain_white}.
At a high level, a blockchain is a simple replicated state machine~\cite{using_time}.
\libra \emph{validators} (replicas) collectively maintain a distributed database that encodes the global state structure described in \Cref{sec:move_overview}.
Users submit transactions to the system that are batched into a \emph{block}, or ordered list of transactions.
The role of the transaction execution layer is to take a block of transactions and the current global state as input and execute each transaction to produce a \emph{write set} representing the effects of the transaction on the global state.
The effect of the block is the ordered composition of the effects of each of its transactions.

The logic for \libra execution lives in two separate places: the \libra adapter, and \libra's \move standard library. The \libra adapter contains about 1K lines of Rust code that wrap the \move virtual machine. The adapter implements logic for splitting a \libra block into transactions, checking a cryptographic signature on the transaction, extracting a \move transaction script,  arguments, and gas budget to pass to the \move virtual machine, and applying the effects of executing the transaction to the storage layer.

\paragraph{\move Standard Library}
\libra's \move standard library consists of 40 modules totalling about 3K lines of \move source code that compile to 44KB of bytecode. Broadly speaking, these modules implement four categories of functionality:
\begin{enumerate}
\item Coins: implementations of both single-currency stablecoins and the multi-currency \code{LBR} coin as described in~\cite{libra_whitepaper}
\item Accounts: several different account types, sequence number logic to prevent replay attacks, sender authentication, key rotation, events for notifying clients
\item Validator management: adding/removing validators, paying gas fees to validators, rotating validator cryptographic keys
\item Utility modules such as \code{Option}, \code{Compare}, and \code{FixedPoint32}
\end{enumerate}

We will present a subset of the \code{Coin} and \code{Account} APIs to give the reader a sense of how \move's resources give us the flexibility to implement our own version of concepts that must be baked into the semantics of other smart contract languages.

\paragraph{\code{Coin} Module}
\begin{figure}
\begin{lstlisting}
resource T { value: u64 }
resource MarketCap { total_value: u64 }

// create a Coin with value=0
fun zero(): Coin::T
// Consume c and increment c_ref by its value
fun deposit(c_ref: &mut Coin::T, c: Coin::T)
// Decrement c_ref by amt, create Coin with value=amt
fun withdraw(c_ref: &mut Coin::T, amt: u64): Coin::T
// create Coin with value=amt, update MarketCap by amt. privileged operation
fun mint(amt: u64): Coin::T
\end{lstlisting}
\caption{A subset of the \code{Coin} module API.}

\label{fig:coin}
\end{figure}

The \code{Coin} module in \Cref{fig:coin} implements the native currency of the \libra platform by wrapping an integer \code{value} with a safe API.
This module clearly illustrates the value of combining linearity with traditional modularity.
Any user can create a coin worth zero, combine two coins with \code{deposit}, or split a single coin into two coins with \code{withdraw}.
The reader might wonder why \code{withdraw} chooses to mutate a \code{&mut Coin::T} rather than expose a functional API that takes two \code{Coin::T}'s and returns a new one.
The answer is that the reference parameter provides needed flexibility for updating a \code{Coin::T} object stored in the field of another resource.
For example: the functional API could not be used to update the \code{balance} field of the \code{Bank::T} resource in \Cref{fig:bank}.

The privileged \code{mint} operation allows a privileged user (the body contains a permission check) to create new currency and update the integer value stored in the \code{MarketCap} resource.
The body of the procedure also ensures that there is a single \code{MarketCap} resource in the system published at address $a$.

The conservation of currency in the \libra system can thus be stated as a local invariant of the \code{Coin} module: the sum of the \code{value}s of each \code{Coin::T} resource in the system must be equal to the \code{total_value} field of the \code{MarketCap} resource.
The combination of the strong encapsulation described in  \Cref{sec:bytecode} and resource safety guarantees defined in \Cref{sec:resource_safety} ensure that \code{Coin::T}'s cannot be created, destroyed, or modified by code outside the \code{Coin} module.
The \code{Coin} module needs only to ensure that \code{deposit} and \code{withdraw} conserve the value fields of the input/output field and that the \code{MarketCap} is updated appropriately whenever new coins are created.
This can be verified with straightforward local reasoning over the \code{Coin} module.

To the best of our knowledge, no other blockchain platform has made a rigorous argument for the conservation of its native currency. We note that there are known counterexamples for conservation such as the Scilla \scilcode{emit} bug described in \Cref{sec:overview}.

\paragraph{\code{Account} Module}
\begin{figure}
\begin{lstlisting}
resource T { bal: Coin::T, seq_num: u64, auth_key: vector<u8>, has_withdraw_cap: bool, ... }
// represents the permission to withdraw from account
resource WithdrawCap { account: address }
fun create(addr: address) // publish a T under addr
// debits the tx sender's balance by amt
fun withdraw_from_sender(amt: u64): Coin::T
// credits recipient's balance by c.value
fun deposit(recipient: address, c: Coin::T)
fun rotate_sender_auth_key(new_auth_key: bytes)
// acquire unique withdraw capability for sender
fun extract_sender_withdraw_cap(): Account::WithdrawCap
// debits the balance at cap.account by amt
fun withdraw(amt: u64, cap: &Account::WithdrawCap): Coin::T
\end{lstlisting}
\caption{A subset of the \code{Account} module API.}

\label{fig:account}
\end{figure}
A \libra user account at address $a$ is represented by storing an \code{Account::T} resource under $a$.
This resource holds all of the information a user needs to transact: a balance, a sequence number to prevent replay attacks, and an authentication key.
The module exposes procedures for withdrawing funds and rotating authentication keys (for the transaction sender only) and depositing funds (to any address).

In addition, the module allows the holder of the \code{WithdrawCap} capability to debit an account (similar to the \code{Bank::Credit} resource in \Cref{fig:bank}).
The implementation of \code{extract_sender_withdraw_cap} (not shown) uses the \code{has_withdraw_cap} field to ensure that there is at most one \code{WithdrawCap} for each account in the system.
An account whose \code{WithdrawCap} has been extracted can no longer use \code{withdraw_from_sender}---using the unique capability for account address $a$ is the only way to debit the balance of $a$.
Similar to native currency conservation in \code{Coin}, the uniqueness property for \code{WithdrawCap} can be established with simple local reasoning in the \code{Account} module.


In addition, using a resource to explicitly represent the permission to withdraw from an account provides significant flexibility for users of \libra.
A common use-case for contracts is placing preconditions on the funds stored in certain addresses; for example:
\begin{itemize}
\item Funds should only be sent to recipients in a whitelist
\item Funds should only be transferred after a certain date
\item Funds should only be withdrawn with the approval of a quorum
\end{itemize}
Each of these policies can be implemented by creating a resource that stores a \code{WithdrawCap} and restricts access accordingly.
Platforms like Ethereum~\cite{ethereum} support this use case by implementing payments with dynamic dispatch and  allowing contracts to override the default payment behavior, but (as we explained in \Cref{sec:overview}), this is a dangerous pattern because payment to an unknown address can call arbitrary code.
The capability-based approach of \code{WithdrawCap} enables custom payment logic without dynamic dispatch by moving the dynamism to the withdrawal code (known and trusted by the sender) instead of the recipient code (unknown and not trusted by the sender).

\paragraph{Deployment in the \libra Testnet}
The \move VM is currently running as part of the public \libra testnet\footnote{\url{https://developers.libra.org/docs/my-first-transaction}} that previews the functionality of the \libra payment system (expected to launch in 2020 pending regulatory approval).
The testnet supports a whitelist of transaction scripts that exercise all of the modules in the \move standard library.
To limit the scope and risk of the launch, the testnet does not currently allow users to publish new modules.
We hope that this will change in time as the \libra Association works with regulators to define appropriate safeguards for third-party publishing of smart contracts.




\subsection{\move Usage Outside of \libra}
\label{sec:move_usage}

The flexibility of the \move language and the modularity of the \move VM has facilitated external interest in/adoption of \move in both academic and industrial contexts.

\paragraph{Other Blockchains}
Solana\footnote{\redacted{\url{https://solana-labs.github.io/book/embedding-move.htm}}} is a multi-language blockchain that supports \move smart contracts and has publicly launched.
The dfinance\footnote{\redacted{\url{https://docs.dfinance.co/move_vm}}} and OpenLibra\footnote{\redacted{\url{https://www.openlibra.io/}}} blockchain platforms are using \move, but have not yet launched.
The Flow blockchain is an upcoming project from Dapper Labs, the creator of the popular CryptoKitties project in Ethereum. Dapper is considering using the Move bytecode as the compilation target for its Cadence source language.\footnote{\redacted{https://medium.com/dapperlabs/libra-and-flow-combining-resources-for-open-source-40530e53fa01}}
PRISM~\cite{yang2019prism} is an academic project that seeks to significantly enhance the scalability of existing blockchain platforms. PRISM has recently implemented smart contract support for both \move and the EVM~\cite{wang2020prism}.

\paragraph{Verification Tools}
The Move Prover~\cite{moveprover} implements a specification language and functional verification tool for \move. It has been used to specify and verify pre- and post-conditions for several of the \libra standard library modules.
A verification startup called Synthetic Minds verified key properties of an earlier version of the \code{Coin}/\code{Account} \libra modules and wrote/verified several new modules.\footnote{\url{https://synthetic-minds.com/pages/blog/blog-2019-09-11.html}}

\section{Related Work}
\label{sec:related_work}

\paragraph{Rust}
Mozilla's Rust~\cite{rust} language is used at large companies such as Google, Amazon, and Facebook.
Rust uses a clever affine type system to provide type and memory safety along with data-race freedom.
\move is strongly influenced by Rust, but there are several important differences:
\begin{enumerate}
\item Affine vs linear: Rust structs can be silently discarded, but \move resources must be explicitly \unpackCmd'ed. This is a profound difference that is required for resource safety, but presents complications in the design of many language features (see \Cref{sec:language_design}).
\item A subset of \move without the persistent global state $\globalstates$ is superficially similar to a linear variant of Rust with many features removed (e.g., references in structs, heap allocation, collections, traits, generics, concurrency, \code{unsafe}). However, the persistent global state of \move gives programmers access to a shared, \emph{mutable} global state. There is no equivalent feature in Rust, and the restrictions of Rust's borrow checker make it impossible to emulate this feature in safe Rust. A key contribution of \move is a representation of global state expressive enough to represent complex smart contracts, yet simple enough to preserve the safety guarantees we desire.
\item Rust is a source language that compiles to an executable representation, whereas \move bytecode is itself an executable representation. The key difference is that the guarantees enforced by the \move language hold directly on the executable representation (no need to trust a compiler) and continue to hold when \move programs are linked against untrusted code. This property is a requirement for smart contracts, which are deployed in the open and must tolerate arbitrary interactions with untrusted code. Thus, even if the Rust language (or a subset) had exactly the properties we wanted, it would not be usable as a smart contract language.
\end{enumerate}

\paragraph{Substructural Type Systems and Ownership Types}
Rust is the most mainstream language with a substructural type system, but it follows in the footsteps of other cleverly designed languages such as Cyclone~\cite{DBLP:conf/pldi/GrossmanMJHWC02}, Clean~\cite{Smetsers:1993:GSD:647364.725667}, Pony~\cite{DBLP:conf/agere/ClebschDBM15}, and Alms~\cite{DBLP:conf/popl/TovP11}.
A related line of work involves ownership type systems~\cite{DBLP:series/lncs/ClarkeOSW13} for controlling aliasing in languages with reference semantics.
A common theme in both areas is leveraging types for safe memory management to avoid undefined behavior due to data races or accessing deallocated/uninitialized memory, but without relying on garbage collection.

\move also builds on this tradition, but our usage of linearity is broader and more ambitious than memory management: linear resources are a natural abstraction for digital money and other programmer-defined assets.
The resource safety guarantee from \Cref{sec:resource_safety} is a novel \emph{semantic} conservation property similar to (e.g.) conservation of mass, stated in a way that is independent from any particular enforcement mechanism (e.g., linear type or dynamic checks).

\paragraph{Linear Logic}
The modern era of substructural logics and associated type systems began with Girard's linear logic~\cite{linear_logic}.
In his early explanation of its resource sensitivity, Girard used an intuitive representation of money.
In that account, each number of fixed-value coins has a different type:
two coins would have type $C \otimes C$ and three coins $C \otimes C \otimes C$.
While this encoding of money illustrates a fundamental difference between linear and intuitionistic logic, the approach resembles giving each individual integer a different type.  (If each integer has a different type, then any straightforward type system would
prohibit a single addition function from being used to add arbitrary pairs of integers.)
Instead of segregating different monetary values in different types,
we need resource types (such as the \code{Coin} module in \Cref{sec:core_modules})
that create coins of any value which cannot be duplicated or destroyed outside the control of their defining module.
Thus, while \move draws on Girard's key insight relating linearity to assets, \move resources give programmers both
provable resource safety \emph{and} flexibility in representing money-like types.

\paragraph{Account-based Blockchain Languages}
\emph{Account-based} blockchain languages mimic a classic bank ledger by representing global state as a map from account addresses to integer balances and exposing language primitives for debiting one balance and crediting another.
We discussed two prominent executable account-based languages~\cite{scilla, ethereum} in \Cref{sec:overview}; many others like IELE~\cite{DBLP:conf/fm/KasampalisGMSZF19}, Agoric JS~\cite{agoric}, Michelson~\cite{michelson}, and Pact~\cite{pact} have been proposed in the past few years.
\move implements a more expressive variant of the account-based model
where account addresses are associated with a direct representation of money
(programmable resources) instead of an indirect one (integer balances).

\paragraph{UTXO-based Blockchain Languages}
\emph{UTXO} (unspent transaction output) blockchain languages represent the global state as a set of (authentication policy, amount) pairs.
Programs transfer money by satisfying the authentication policy of one or more input UTXOs and creating a set of fresh output UTXOs whose amounts sum to the amounts included in the inputs.
Program execution removes the input UTXOs from the state and adds the fresh ones.
This model was pioneered by Bitcoin's~\cite{nakamoto} Script~\cite{bitcoin_script} language, and has been adopted by a few more recent languages such as Simplicity~\cite{simplicity} and Plutus~\cite{plutus}.
Though UTXOs are a good choice for a platform with a single native currency and limited programmability, they are cumbersome to use for general-purpose state changes. We feel that \move's resources are a more flexible approach to implementing diverse financial assets with customizable behaviors.

\paragraph{Blockchain Source Languages with Linear Types}
Flint~\cite{DBLP:journals/corr/abs-1904-06534} and Obsidian~\cite{DBLP:conf/icse/Coblenz17,DBLP:journals/corr/abs-1909-03523} are contract programming languages that use linear types as an explicit representation of assets.
These languages enforce linearity at the source level, but compile to an executable representation (EVM bytecode~\cite{ethereum} and a Java subset used by Hyperledger Fabric~\cite{hyperledger}, respectively) without the same protections.
Nomos~\cite{das2019resourceaware} uses session types to achieve even stronger static protections, but also does not consider the problem of applying the type system to an executable blockchain representation.
The Cadence source language\footnote{\url{https://docs.onflow.org/docs/cadence}} has linear types and is considering \move bytecode as a compilation target (see \Cref{sec:move_usage}).

The distinguishing feature of \move is an executable bytecode representation with resource safety guarantees for \emph{all} programs.
This is crucially important given the open deployment model for contracts---recall that any contract must tolerate arbitrary interactions with untrusted code.
Source-level linearity has limited value if it can be violated by untrusted code at the executable level (e.g., untrusted code that duplicates a source-level linear type).
\section{Conclusion}
\label{sec:conclusion}

We have introduced language support for a specific form of linear resource types, formalized corresponding semantic resource safety properties,
and proved that they hold for successful concrete execution of \move programs.
The language construct provides a safe abstraction for currency-like values in the \move language, as illustrated by example in \Cref{sec:overview}
and through more extensive experience summarized in \Cref{sec:experience}.
In future work, we plan to describe and formalize the \move bytecode verifier, which involves  interesting and novel static analyses for ensuring type, resource, and reference safety invariants.
One goal of the verifier is to ensure progress for concrete execution of \move programs in the checking semantics of the present paper, complementing the safety guarantees proved here.
Overall, we believe the semantic guarantees and successful programming experience presented in this paper suggest that the language design and implementation provide
better  language support and more effective design patterns for the growing range of resource-sensitive applications of blockchain and related platforms.

\bibliography{bibfile}


\begin{thebibliography}{42}


\ifx \showCODEN    \undefined \def \showCODEN     #1{\unskip}     \fi
\ifx \showDOI      \undefined \def \showDOI       #1{#1}\fi
\ifx \showISBNx    \undefined \def \showISBNx     #1{\unskip}     \fi
\ifx \showISBNxiii \undefined \def \showISBNxiii  #1{\unskip}     \fi
\ifx \showISSN     \undefined \def \showISSN      #1{\unskip}     \fi
\ifx \showLCCN     \undefined \def \showLCCN      #1{\unskip}     \fi
\ifx \shownote     \undefined \def \shownote      #1{#1}          \fi
\ifx \showarticletitle \undefined \def \showarticletitle #1{#1}   \fi
\ifx \showURL      \undefined \def \showURL       {\relax}        \fi
\providecommand\bibfield[2]{#2}
\providecommand\bibinfo[2]{#2}
\providecommand\natexlab[1]{#1}
\providecommand\showeprint[2][]{arXiv:#2}

\bibitem[\protect\citeauthoryear{Agoric}{Agoric}{2019}]%
        {agoric}
\bibfield{author}{\bibinfo{person}{Agoric}.} \bibinfo{year}{2019}\natexlab{}.
\newblock
\newblock
\urldef\tempurl%
\url{https://github.com/Agoric/SES}
\showURL{%
\tempurl}


\bibitem[\protect\citeauthoryear{Amsden, Arora, Bano, Baudet, Blackshear,
  Bothra, Cabrera, Catalini, Chalkias, Cheng, Ching, Chursin, Danezis, Giacomo,
  Dill, Ding, Doudchenko, Gao, Gao, Garillot, Gorven, Hayes, Hou, Hu, Hurley,
  Lewi, Li, Li, Malkhi, Margulis, Maurer, Mohassel, de~Naurois, Nikolaenko,
  Nowacki, Orlov, Perelman, Pott, Proctor, Qadeer, Rain, Russi, Schwab, Sezer,
  Sonnino, Venter, Wei, Wernerfelt, Williams, Wu, Yan, Zakian, and Zhou}{Amsden
  et~al\mbox{.}}{2019}]%
        {libra_blockchain_white}
\bibfield{author}{\bibinfo{person}{Zachary Amsden}, \bibinfo{person}{Ramnik
  Arora}, \bibinfo{person}{Shehar Bano}, \bibinfo{person}{Mathieu Baudet},
  \bibinfo{person}{Sam Blackshear}, \bibinfo{person}{Abhay Bothra},
  \bibinfo{person}{George Cabrera}, \bibinfo{person}{Christian Catalini},
  \bibinfo{person}{Konstantinos Chalkias}, \bibinfo{person}{Evan Cheng},
  \bibinfo{person}{Avery Ching}, \bibinfo{person}{Andrey Chursin},
  \bibinfo{person}{George Danezis}, \bibinfo{person}{Gerardo~Di Giacomo},
  \bibinfo{person}{David~L. Dill}, \bibinfo{person}{Hui Ding},
  \bibinfo{person}{Nick Doudchenko}, \bibinfo{person}{Victor Gao},
  \bibinfo{person}{Zhenhuan Gao}, \bibinfo{person}{François Garillot},
  \bibinfo{person}{Michael Gorven}, \bibinfo{person}{Philip Hayes},
  \bibinfo{person}{J.~Mark Hou}, \bibinfo{person}{Yuxuan Hu},
  \bibinfo{person}{Kevin Hurley}, \bibinfo{person}{Kevin Lewi},
  \bibinfo{person}{Chunqi Li}, \bibinfo{person}{Zekun Li},
  \bibinfo{person}{Dahlia Malkhi}, \bibinfo{person}{Sonia Margulis},
  \bibinfo{person}{Ben Maurer}, \bibinfo{person}{Payman Mohassel},
  \bibinfo{person}{Ladi de Naurois}, \bibinfo{person}{Valeria Nikolaenko},
  \bibinfo{person}{Todd Nowacki}, \bibinfo{person}{Oleksandr Orlov},
  \bibinfo{person}{Dmitri Perelman}, \bibinfo{person}{Alistair Pott},
  \bibinfo{person}{Brett Proctor}, \bibinfo{person}{Shaz Qadeer},
  \bibinfo{person}{Rain}, \bibinfo{person}{Dario Russi}, \bibinfo{person}{Bryan
  Schwab}, \bibinfo{person}{Stephane Sezer}, \bibinfo{person}{Alberto Sonnino},
  \bibinfo{person}{Herman Venter}, \bibinfo{person}{Lei Wei},
  \bibinfo{person}{Nils Wernerfelt}, \bibinfo{person}{Brandon Williams},
  \bibinfo{person}{Qinfan Wu}, \bibinfo{person}{Xifan Yan},
  \bibinfo{person}{Tim Zakian}, {and} \bibinfo{person}{Runtian Zhou}.}
  \bibinfo{year}{2019}\natexlab{}.
\newblock \bibinfo{title}{The {L}ibra {B}lockchain}.
\newblock
  \bibinfo{howpublished}{\url{https://developers.libra.org/docs/the-libra-blockchain-paper}}.
\newblock


\bibitem[\protect\citeauthoryear{Armstrong}{Armstrong}{2019}]%
        {JPMcoin}
\bibfield{author}{\bibinfo{person}{R. Armstrong}.}
  \bibinfo{year}{2019}\natexlab{}.
\newblock \bibinfo{title}{JPMorgan plan to coin it on the blockchain}.
\newblock \bibinfo{howpublished}{Financial Times}.
\newblock
\urldef\tempurl%
\url{https://www.ft.com/content/0bafd9d6-307e-11e9-8744-e7016697f225}
\showURL{%
\tempurl}


\bibitem[\protect\citeauthoryear{Atzei, Bartoletti, and Cimoli}{Atzei
  et~al\mbox{.}}{2017}]%
        {eth_vulns}
\bibfield{author}{\bibinfo{person}{Nicola Atzei}, \bibinfo{person}{Massimo
  Bartoletti}, {and} \bibinfo{person}{Tiziana Cimoli}.}
  \bibinfo{year}{2017}\natexlab{}.
\newblock \showarticletitle{A Survey of Attacks on Ethereum Smart Contracts
  (SoK)}. In \bibinfo{booktitle}{\emph{Principles of Security and Trust - 6th
  International Conference, {POST} 2017, Held as Part of the European Joint
  Conferences on Theory and Practice of Software, {ETAPS} 2017, Uppsala,
  Sweden, April 22-29, 2017, Proceedings}}. \bibinfo{pages}{164--186}.
\newblock


\bibitem[\protect\citeauthoryear{Blackshear, Cheng, Dill, Gao, Maurer, Nowacki,
  Pott, Qadeer, Rain, Russi, Sezer, Zakian, and Zhou}{Blackshear
  et~al\mbox{.}}{2019}]%
        {move_white}
\bibfield{author}{\bibinfo{person}{Sam Blackshear}, \bibinfo{person}{Evan
  Cheng}, \bibinfo{person}{David~L. Dill}, \bibinfo{person}{Victor Gao},
  \bibinfo{person}{Ben Maurer}, \bibinfo{person}{Todd Nowacki},
  \bibinfo{person}{Alistair Pott}, \bibinfo{person}{Shaz Qadeer},
  \bibinfo{person}{Rain}, \bibinfo{person}{Dario Russi},
  \bibinfo{person}{Stephane Sezer}, \bibinfo{person}{Tim Zakian}, {and}
  \bibinfo{person}{Runtian Zhou}.} \bibinfo{year}{2019}\natexlab{}.
\newblock \bibinfo{title}{Move: A Language With Programmable Resources}.
\newblock
  \bibinfo{howpublished}{\url{https://developers.libra.org/docs/move-paper}}.
\newblock


\bibitem[\protect\citeauthoryear{Buterin}{Buterin}{2016}]%
        {re_dao}
\bibfield{author}{\bibinfo{person}{Vitalik Buterin}.}
  \bibinfo{year}{2016}\natexlab{}.
\newblock \bibinfo{title}{Critical update re {DAO}}.
\newblock
\newblock
\urldef\tempurl%
\url{https://ethereum.github.io/blog/2016/06/17/critical-update-re-dao-vulnerability}
\showURL{%
\tempurl}


\bibitem[\protect\citeauthoryear{Casey and Wong}{Casey and Wong}{2017}]%
        {HBR-global-supply-chain}
\bibfield{author}{\bibinfo{person}{M.J. Casey} {and} \bibinfo{person}{P.
  Wong}.} \bibinfo{year}{2017}\natexlab{}.
\newblock \bibinfo{title}{Global Supply Chains Are About to Get Better, Thanks
  to Blockchain}.
\newblock \bibinfo{howpublished}{Harvard Business Review}.
\newblock


\bibitem[\protect\citeauthoryear{Clarke, {\"{O}}stlund, Sergey, and
  Wrigstad}{Clarke et~al\mbox{.}}{2013}]%
        {DBLP:series/lncs/ClarkeOSW13}
\bibfield{author}{\bibinfo{person}{Dave Clarke}, \bibinfo{person}{Johan
  {\"{O}}stlund}, \bibinfo{person}{Ilya Sergey}, {and} \bibinfo{person}{Tobias
  Wrigstad}.} \bibinfo{year}{2013}\natexlab{}.
\newblock \showarticletitle{Ownership Types: {A} Survey}.
\newblock In \bibinfo{booktitle}{\emph{Aliasing in Object-Oriented Programming.
  Types, Analysis and Verification}}. \bibinfo{pages}{15--58}.
\newblock
\urldef\tempurl%
\url{https://doi.org/10.1007/978-3-642-36946-9\_3}
\showDOI{\tempurl}


\bibitem[\protect\citeauthoryear{Clebsch, Drossopoulou, Blessing, and
  McNeil}{Clebsch et~al\mbox{.}}{2015}]%
        {DBLP:conf/agere/ClebschDBM15}
\bibfield{author}{\bibinfo{person}{Sylvan Clebsch}, \bibinfo{person}{Sophia
  Drossopoulou}, \bibinfo{person}{Sebastian Blessing}, {and}
  \bibinfo{person}{Andy McNeil}.} \bibinfo{year}{2015}\natexlab{}.
\newblock \showarticletitle{Deny capabilities for safe, fast actors}. In
  \bibinfo{booktitle}{\emph{Proceedings of the 5th International Workshop on
  Programming Based on Actors, Agents, and Decentralized Control, AGERE! 2015,
  Pittsburgh, PA, USA, October 26, 2015}}. \bibinfo{pages}{1--12}.
\newblock
\urldef\tempurl%
\url{https://doi.org/10.1145/2824815.2824816}
\showDOI{\tempurl}


\bibitem[\protect\citeauthoryear{Coblenz}{Coblenz}{2017}]%
        {DBLP:conf/icse/Coblenz17}
\bibfield{author}{\bibinfo{person}{Michael~J. Coblenz}.}
  \bibinfo{year}{2017}\natexlab{}.
\newblock \showarticletitle{Obsidian: a safer blockchain programming language}.
  In \bibinfo{booktitle}{\emph{Proceedings of the 39th International Conference
  on Software Engineering, {ICSE} 2017, Buenos Aires, Argentina, May 20-28,
  2017 - Companion Volume}}. \bibinfo{pages}{97--99}.
\newblock
\urldef\tempurl%
\url{https://doi.org/10.1109/ICSE-C.2017.150}
\showDOI{\tempurl}


\bibitem[\protect\citeauthoryear{Coblenz, Oei, Etzel, Koronkevich, Baker,
  Bloem, Myers, Sunshine, and Aldrich}{Coblenz et~al\mbox{.}}{2019}]%
        {DBLP:journals/corr/abs-1909-03523}
\bibfield{author}{\bibinfo{person}{Michael~J. Coblenz}, \bibinfo{person}{Reed
  Oei}, \bibinfo{person}{Tyler Etzel}, \bibinfo{person}{Paulette Koronkevich},
  \bibinfo{person}{Miles Baker}, \bibinfo{person}{Yannick Bloem},
  \bibinfo{person}{Brad~A. Myers}, \bibinfo{person}{Joshua Sunshine}, {and}
  \bibinfo{person}{Jonathan Aldrich}.} \bibinfo{year}{2019}\natexlab{}.
\newblock \showarticletitle{Obsidian: Typestate and Assets for Safer Blockchain
  Programming}.
\newblock \bibinfo{journal}{\emph{CoRR}}  \bibinfo{volume}{abs/1909.03523}
  (\bibinfo{year}{2019}).
\newblock
\showeprint[arxiv]{1909.03523}
\urldef\tempurl%
\url{http://arxiv.org/abs/1909.03523}
\showURL{%
\tempurl}


\bibitem[\protect\citeauthoryear{Das, Balzer, Hoffmann, and Pfenning}{Das
  et~al\mbox{.}}{2019}]%
        {das2019resourceaware}
\bibfield{author}{\bibinfo{person}{Ankush Das}, \bibinfo{person}{Stephanie
  Balzer}, \bibinfo{person}{Jan Hoffmann}, {and} \bibinfo{person}{Frank
  Pfenning}.} \bibinfo{year}{2019}\natexlab{}.
\newblock \bibinfo{title}{Resource-Aware Session Types for Digital Contracts}.
\newblock
\newblock
\showeprint[arxiv]{cs.PL/1902.06056}


\bibitem[\protect\citeauthoryear{Foundation}{Foundation}{2018a}]%
        {solidity}
\bibfield{author}{\bibinfo{person}{Ethereum Foundation}.}
  \bibinfo{year}{2018}\natexlab{a}.
\newblock \bibinfo{title}{Solidity Documentation}.
\newblock
\newblock
\urldef\tempurl%
\url{http://solidity.readthedocs.io}
\showURL{%
\tempurl}


\bibitem[\protect\citeauthoryear{Foundation}{Foundation}{2019}]%
        {plutus}
\bibfield{author}{\bibinfo{person}{IOHK Foundation}.}
  \bibinfo{year}{2019}\natexlab{}.
\newblock \bibinfo{title}{Plutus: A Functional Contract Platform}.
\newblock \bibinfo{howpublished}{\url{https://testnet.iohkdev.io/plutus}}.
\newblock


\bibitem[\protect\citeauthoryear{Foundation}{Foundation}{2018b}]%
        {michelson}
\bibfield{author}{\bibinfo{person}{Tezos Foundation}.}
  \bibinfo{year}{2018}\natexlab{b}.
\newblock \bibinfo{title}{Michelson: the language of Smart Contracts in Tezos}.
\newblock
  \bibinfo{howpublished}{\url{http://tezos.gitlab.io/mainnet/whitedoc/michelson.html}}.
\newblock


\bibitem[\protect\citeauthoryear{Foundation}{Foundation}{2018c}]%
        {hyperledger}
\bibfield{author}{\bibinfo{person}{The~Linux Foundation}.}
  \bibinfo{year}{2018}\natexlab{c}.
\newblock \bibinfo{title}{Hyperledger Fabric}.
\newblock
  \bibinfo{howpublished}{\url{https://www.hyperledger.org/projects/fabric}}.
\newblock


\bibitem[\protect\citeauthoryear{Girard}{Girard}{1987}]%
        {linear_logic}
\bibfield{author}{\bibinfo{person}{Jean{-}Yves Girard}.}
  \bibinfo{year}{1987}\natexlab{}.
\newblock \showarticletitle{Linear Logic}.
\newblock \bibinfo{journal}{\emph{Theor. Comput. Sci.}} (\bibinfo{year}{1987}).
\newblock


\bibitem[\protect\citeauthoryear{Grossman, Morrisett, Jim, Hicks, Wang, and
  Cheney}{Grossman et~al\mbox{.}}{2002}]%
        {DBLP:conf/pldi/GrossmanMJHWC02}
\bibfield{author}{\bibinfo{person}{Dan Grossman}, \bibinfo{person}{J.~Gregory
  Morrisett}, \bibinfo{person}{Trevor Jim}, \bibinfo{person}{Michael~W. Hicks},
  \bibinfo{person}{Yanling Wang}, {and} \bibinfo{person}{James Cheney}.}
  \bibinfo{year}{2002}\natexlab{}.
\newblock \showarticletitle{Region-Based Memory Management in Cyclone}. In
  \bibinfo{booktitle}{\emph{Proceedings of the 2002 {ACM} {SIGPLAN} Conference
  on Programming Language Design and Implementation (PLDI), Berlin, Germany,
  June 17-19, 2002}}. \bibinfo{pages}{282--293}.
\newblock
\urldef\tempurl%
\url{https://doi.org/10.1145/512529.512563}
\showDOI{\tempurl}


\bibitem[\protect\citeauthoryear{Hardy}{Hardy}{1994}]%
        {Hardy94theconfused}
\bibfield{author}{\bibinfo{person}{Norman Hardy}.}
  \bibinfo{year}{1994}\natexlab{}.
\newblock \showarticletitle{The Confused Deputy (or why capabilities might have
  been invented)}.
\newblock \bibinfo{journal}{\emph{ACM SIGOPS Operating Systems Review}}
  \bibinfo{volume}{22} (\bibinfo{year}{1994}), \bibinfo{pages}{36--38}.
\newblock


\bibitem[\protect\citeauthoryear{Kasampalis, Guth, Moore, Serbanuta, Zhang,
  Filaretti, Serbanuta, Johnson, and Rosu}{Kasampalis et~al\mbox{.}}{2019}]%
        {DBLP:conf/fm/KasampalisGMSZF19}
\bibfield{author}{\bibinfo{person}{Theodoros Kasampalis},
  \bibinfo{person}{Dwight Guth}, \bibinfo{person}{Brandon~M. Moore},
  \bibinfo{person}{Traian{-}Florin Serbanuta}, \bibinfo{person}{Yi Zhang},
  \bibinfo{person}{Daniele Filaretti}, \bibinfo{person}{Virgil~Nicolae
  Serbanuta}, \bibinfo{person}{Ralph Johnson}, {and} \bibinfo{person}{Grigore
  Rosu}.} \bibinfo{year}{2019}\natexlab{}.
\newblock \showarticletitle{{IELE:} {A} Rigorously Designed Language and Tool
  Ecosystem for the Blockchain}. In \bibinfo{booktitle}{\emph{Formal Methods -
  The Next 30 Years - Third World Congress, {FM} 2019, Porto, Portugal, October
  7-11, 2019, Proceedings}}. \bibinfo{pages}{593--610}.
\newblock


\bibitem[\protect\citeauthoryear{Lamport}{Lamport}{1984}]%
        {using_time}
\bibfield{author}{\bibinfo{person}{Leslie Lamport}.}
  \bibinfo{year}{1984}\natexlab{}.
\newblock \showarticletitle{Using Time Instead of Timeout for Fault-Tolerant
  Distributed Systems}.
\newblock \bibinfo{journal}{\emph{{ACM} Trans. Program. Lang. Syst.}}
  (\bibinfo{year}{1984}).
\newblock


\bibitem[\protect\citeauthoryear{Lindholm and Yellin}{Lindholm and
  Yellin}{1997}]%
        {jvm}
\bibfield{author}{\bibinfo{person}{Tim Lindholm} {and} \bibinfo{person}{Frank
  Yellin}.} \bibinfo{year}{1997}\natexlab{}.
\newblock \bibinfo{booktitle}{\emph{The {J}ava Virtual Machine Specification}}.
\newblock \bibinfo{publisher}{Addison-Wesley}.
\newblock


\bibitem[\protect\citeauthoryear{Matsakis and Klock}{Matsakis and
  Klock}{2014}]%
        {rust}
\bibfield{author}{\bibinfo{person}{Nicholas~D. Matsakis} {and}
  \bibinfo{person}{Felix~S. Klock, II}.} \bibinfo{year}{2014}\natexlab{}.
\newblock \showarticletitle{The Rust Language}.
\newblock \bibinfo{journal}{\emph{Ada Lett.}} \bibinfo{volume}{34},
  \bibinfo{number}{3} (\bibinfo{date}{Oct.} \bibinfo{year}{2014}),
  \bibinfo{pages}{103--104}.
\newblock
\showISSN{1094-3641}
\urldef\tempurl%
\url{https://doi.org/10.1145/2692956.2663188}
\showDOI{\tempurl}


\bibitem[\protect\citeauthoryear{Meijer, Wa, and Gough}{Meijer
  et~al\mbox{.}}{2000}]%
        {clr}
\bibfield{author}{\bibinfo{person}{Erik Meijer}, \bibinfo{person}{Redmond Wa},
  {and} \bibinfo{person}{John Gough}.} \bibinfo{year}{2000}\natexlab{}.
\newblock \bibinfo{title}{Technical Overview of the Common Language Runtime}.
\newblock
\newblock


\bibitem[\protect\citeauthoryear{Miller, Yee, and Shapiro}{Miller
  et~al\mbox{.}}{2003}]%
        {capability-myths}
\bibfield{author}{\bibinfo{person}{Mark~S. Miller}, \bibinfo{person}{Ka-Ping
  Yee}, {and} \bibinfo{person}{Jonathan Shapiro}.}
  \bibinfo{year}{2003}\natexlab{}.
\newblock \bibinfo{booktitle}{\emph{Capability Myths Demolished}}.
\newblock \bibinfo{type}{{T}echnical {R}eport}.
\newblock


\bibitem[\protect\citeauthoryear{Nakamoto}{Nakamoto}{2008}]%
        {nakamoto}
\bibfield{author}{\bibinfo{person}{Satoshi Nakamoto}.}
  \bibinfo{year}{2008}\natexlab{}.
\newblock \showarticletitle{Bitcoin: A peer-to-peer electronic cash system}.
\newblock  (\bibinfo{year}{2008}).
\newblock
\urldef\tempurl%
\url{http://bitcoin.org/bitcoin.pdf}
\showURL{%
\tempurl}


\bibitem[\protect\citeauthoryear{O’Connor}{O’Connor}{2017}]%
        {simplicity}
\bibfield{author}{\bibinfo{person}{Russell O’Connor}.}
  \bibinfo{year}{2017}\natexlab{}.
\newblock \bibinfo{title}{Simplicity: A New Language for Blockchains}.
\newblock
\newblock
\urldef\tempurl%
\url{https://blockstream.com/simplicity.pdf}
\showURL{%
\tempurl}


\bibitem[\protect\citeauthoryear{Palladino}{Palladino}{2017}]%
        {parity_hack}
\bibfield{author}{\bibinfo{person}{Santiago Palladino}.}
  \bibinfo{year}{2017}\natexlab{}.
\newblock \bibinfo{title}{The {P}arity wallet hack explained}.
\newblock
\newblock
\urldef\tempurl%
\url{https://blog. zeppelin.
  solutions/on-the-parity-wallet-multisig-hack-405a8c12e8f7}
\showURL{%
\tempurl}


\bibitem[\protect\citeauthoryear{Popejoy}{Popejoy}{2017}]%
        {pact}
\bibfield{author}{\bibinfo{person}{Stuart Popejoy}.}
  \bibinfo{year}{2017}\natexlab{}.
\newblock \bibinfo{title}{The Pact Smart-Contract Language}.
\newblock \bibinfo{howpublished}{\url{
  http://kadena.io/docs/Kadena-PactWhitepaper.pdf}}.
\newblock


\bibitem[\protect\citeauthoryear{Schrans, Hails, Harkness, Drossopoulou, and
  Eisenbach}{Schrans et~al\mbox{.}}{2019}]%
        {DBLP:journals/corr/abs-1904-06534}
\bibfield{author}{\bibinfo{person}{Franklin Schrans}, \bibinfo{person}{Daniel
  Hails}, \bibinfo{person}{Alexander Harkness}, \bibinfo{person}{Sophia
  Drossopoulou}, {and} \bibinfo{person}{Susan Eisenbach}.}
  \bibinfo{year}{2019}\natexlab{}.
\newblock \showarticletitle{Flint for Safer Smart Contracts}.
\newblock \bibinfo{journal}{\emph{CoRR}}  \bibinfo{volume}{abs/1904.06534}
  (\bibinfo{year}{2019}).
\newblock
\showeprint[arxiv]{1904.06534}
\urldef\tempurl%
\url{http://arxiv.org/abs/1904.06534}
\showURL{%
\tempurl}


\bibitem[\protect\citeauthoryear{Sergey, Nagaraj, Johannsen, Kumar, Trunov, and
  Hao}{Sergey et~al\mbox{.}}{2019}]%
        {scilla}
\bibfield{author}{\bibinfo{person}{Ilya Sergey}, \bibinfo{person}{Vaivaswatha
  Nagaraj}, \bibinfo{person}{Jacob Johannsen}, \bibinfo{person}{Amrit Kumar},
  \bibinfo{person}{Anton Trunov}, {and} \bibinfo{person}{Ken Chan~Guan Hao}.}
  \bibinfo{year}{2019}\natexlab{}.
\newblock \showarticletitle{Safer smart contract programming with Scilla}.
\newblock \bibinfo{journal}{\emph{{PACMPL}}} \bibinfo{volume}{3},
  \bibinfo{number}{{OOPSLA}} (\bibinfo{year}{2019}),
  \bibinfo{pages}{185:1--185:30}.
\newblock
\urldef\tempurl%
\url{https://doi.org/10.1145/3360611}
\showDOI{\tempurl}


\bibitem[\protect\citeauthoryear{Smetsers, Barendsen, Eekelen, and
  Plasmeijer}{Smetsers et~al\mbox{.}}{1994}]%
        {Smetsers:1993:GSD:647364.725667}
\bibfield{author}{\bibinfo{person}{Sjaak Smetsers}, \bibinfo{person}{Erik
  Barendsen}, \bibinfo{person}{Marko C. J. D.~van Eekelen}, {and}
  \bibinfo{person}{Marinus~J. Plasmeijer}.} \bibinfo{year}{1994}\natexlab{}.
\newblock \showarticletitle{Guaranteeing Safe Destructive Updates Through a
  Type System with Uniqueness Information for Graphs}. In
  \bibinfo{booktitle}{\emph{Proceedings of the International Workshop on Graph
  Transformations in Computer Science}}. \bibinfo{publisher}{Springer-Verlag},
  \bibinfo{address}{Berlin, Heidelberg}, \bibinfo{pages}{358--379}.
\newblock
\showISBNx{3-540-57787-4}
\urldef\tempurl%
\url{http://dl.acm.org/citation.cfm?id=647364.725667}
\showURL{%
\tempurl}


\bibitem[\protect\citeauthoryear{Stroustrup}{Stroustrup}{2013}]%
        {c++}
\bibfield{author}{\bibinfo{person}{Bjarne Stroustrup}.}
  \bibinfo{year}{2013}\natexlab{}.
\newblock \bibinfo{booktitle}{\emph{The C++ Programming Language}
  (\bibinfo{edition}{4th} ed.)}.
\newblock \bibinfo{publisher}{Addison-Wesley Professional}.
\newblock
\showISBNx{0321563840, 9780321563842}


\bibitem[\protect\citeauthoryear{Swasey, Garg, and Dreyer}{Swasey
  et~al\mbox{.}}{2017}]%
        {DBLP:journals/pacmpl/Swasey0D17}
\bibfield{author}{\bibinfo{person}{David Swasey}, \bibinfo{person}{Deepak
  Garg}, {and} \bibinfo{person}{Derek Dreyer}.}
  \bibinfo{year}{2017}\natexlab{}.
\newblock \showarticletitle{Robust and compositional verification of object
  capability patterns}.
\newblock \bibinfo{journal}{\emph{{PACMPL}}} \bibinfo{volume}{1},
  \bibinfo{number}{{OOPSLA}} (\bibinfo{year}{2017}),
  \bibinfo{pages}{89:1--89:26}.
\newblock
\urldef\tempurl%
\url{https://doi.org/10.1145/3133913}
\showDOI{\tempurl}


\bibitem[\protect\citeauthoryear{Szabo}{Szabo}{1997}]%
        {szabo_smart_contracts}
\bibfield{author}{\bibinfo{person}{Nick Szabo}.}
  \bibinfo{year}{1997}\natexlab{}.
\newblock \showarticletitle{Formalizing and Securing Relationships on Public
  Networks}.
\newblock \bibinfo{journal}{\emph{First Monday}} \bibinfo{volume}{2},
  \bibinfo{number}{9} (\bibinfo{year}{1997}).
\newblock
\urldef\tempurl%
\url{https://ojphi.org/ojs/index.php/fm/article/view/548}
\showURL{%
\tempurl}


\bibitem[\protect\citeauthoryear{{The Libra Association}}{{The Libra
  Association}}{2020}]%
        {libra_whitepaper}
\bibfield{author}{\bibinfo{person}{{The Libra Association}}.}
  \bibinfo{year}{2020}\natexlab{}.
\newblock \bibinfo{title}{The {L}ibra {W}hite {P}aper v2.0}.
\newblock \bibinfo{howpublished}{\url{https://libra.org/en-us/white-paper}}.
\newblock


\bibitem[\protect\citeauthoryear{Tov and Pucella}{Tov and Pucella}{2011}]%
        {DBLP:conf/popl/TovP11}
\bibfield{author}{\bibinfo{person}{Jesse~A. Tov} {and}
  \bibinfo{person}{Riccardo Pucella}.} \bibinfo{year}{2011}\natexlab{}.
\newblock \showarticletitle{Practical affine types}. In
  \bibinfo{booktitle}{\emph{Proceedings of the 38th {ACM} {SIGPLAN-SIGACT}
  Symposium on Principles of Programming Languages, {POPL} 2011, Austin, TX,
  USA, January 26-28, 2011}}. \bibinfo{pages}{447--458}.
\newblock
\urldef\tempurl%
\url{https://doi.org/10.1145/1926385.1926436}
\showDOI{\tempurl}


\bibitem[\protect\citeauthoryear{Wang, Wang, Bagaria, Tse, and Viswanath}{Wang
  et~al\mbox{.}}{2020}]%
        {wang2020prism}
\bibfield{author}{\bibinfo{person}{Gerui Wang}, \bibinfo{person}{Shuo Wang},
  \bibinfo{person}{Vivek Bagaria}, \bibinfo{person}{David Tse}, {and}
  \bibinfo{person}{Pramod Viswanath}.} \bibinfo{year}{2020}\natexlab{}.
\newblock \bibinfo{title}{Prism Removes Consensus Bottleneck for Smart
  Contracts}.
\newblock
\newblock
\showeprint[arxiv]{cs.DC/2004.08776}


\bibitem[\protect\citeauthoryear{Wiki}{Wiki}{line}]%
        {bitcoin_script}
\bibfield{author}{\bibinfo{person}{Bitcoin Wiki}.}
  \bibinfo{year}{online}\natexlab{}.
\newblock \bibinfo{title}{Bitcoin Script}.
\newblock
\newblock
\urldef\tempurl%
\url{https://en.bitcoin.it/wiki/Script}
\showURL{%
\tempurl}


\bibitem[\protect\citeauthoryear{Wood}{Wood}{2014}]%
        {ethereum}
\bibfield{author}{\bibinfo{person}{Gavin Wood}.}
  \bibinfo{year}{2014}\natexlab{}.
\newblock \showarticletitle{Ethereum: A secure decentralised generalised
  transaction ledger}.
\newblock  (\bibinfo{year}{2014}).
\newblock
\urldef\tempurl%
\url{https://ethereum.github.io/yellowpaper/paper.pdf}
\showURL{%
\tempurl}


\bibitem[\protect\citeauthoryear{Yang, Bagaria, Wang, Alizadeh, Tse, Fanti, and
  Viswanath}{Yang et~al\mbox{.}}{2019}]%
        {yang2019prism}
\bibfield{author}{\bibinfo{person}{Lei Yang}, \bibinfo{person}{Vivek Bagaria},
  \bibinfo{person}{Gerui Wang}, \bibinfo{person}{Mohammad Alizadeh},
  \bibinfo{person}{David Tse}, \bibinfo{person}{Giulia Fanti}, {and}
  \bibinfo{person}{Pramod Viswanath}.} \bibinfo{year}{2019}\natexlab{}.
\newblock \bibinfo{title}{Prism: Scaling Bitcoin by 10,000x}.
\newblock
\newblock
\showeprint[arxiv]{cs.DC/1909.11261}


\bibitem[\protect\citeauthoryear{Zhong, Cheang, Qadeer, Grieskamp, Blackshear,
  Zohar, Barrett, and Dill}{Zhong et~al\mbox{.}}{2020}]%
        {moveprover}
\bibfield{author}{\bibinfo{person}{Jingyi~Emma Zhong}, \bibinfo{person}{Kevin
  Cheang}, \bibinfo{person}{Shaz Qadeer}, \bibinfo{person}{Wolfgang Grieskamp},
  \bibinfo{person}{Sam Blackshear}, \bibinfo{person}{Yoni Zohar},
  \bibinfo{person}{Clark Barrett}, {and} \bibinfo{person}{David~L. Dill}.}
  \bibinfo{year}{2020}\natexlab{}.
\newblock \bibinfo{title}{The Move Prover}.
\newblock
\newblock
\newblock
\shownote{To appear in CAV 2020.}


\end{thebibliography}


\end{document}